\documentclass{acmsmall} 
\runningfoot{}
\gdef\firstfoot{}

\usepackage{tikz} \usetikzlibrary{arrows}
\usepackage{adjustbox}
\usepackage{fullpage}
\usepackage{amssymb} 

\usepackage{versions}
\includeversion{EXAMPLE}\excludeversion{INUTILE}
\excludeversion{JEREMY}\excludeversion{PEYMAN}\excludeversion{TIMOTHY}
\includeversion{LONG}\excludeversion{SHORT}

\newcommand{\A}{\mathcal{A}}

\newcommand{\R}{\mathbb{R}}
\newcommand{\D}{\Delta}
\newcommand{\polylog}{\,{\rm polylog}\,}
\newcommand{\down}[1]{\left\lfloor #1\right\rfloor}
\newcommand{\up}[1]{\left\lceil #1\right\rceil}
\newcommand{\eps}{\varepsilon}
\newcommand{\de}{\delta}
\newcommand{\PiVert}{\Pi_{\mbox{\scriptsize vert}}}
\newcommand{\Ex}{\mathbb{E}}
\newtheorem{observation}[theorem]{Observation}
\newenvironment{algm}{\par\smallskip\noindent\begin{center}\begin{algmbody}}%
  {\end{algmbody}\end{center}\par\smallskip}
  {\par\bigskip\noindent\hspace*{0.1\textwidth}\begin{algmbody}}%
  {\end{algmbody}\par\bigskip}
\newcommand{\algmindent}{19.M}
\newenvironment{algmbody}%
  {\begin{minipage}[t]{\textwidth}\begin{tabbing}%
   \algmindent\= for\= for\= for\= for\= for\=\+\kill}%
  {\end{tabbing}\end{minipage}}
\newcommand{\TO}{,\ldots,}
\providecommand{\entropy}{{\ensuremath\cal H}}
\providecommand{\etal}{{et al.}}

\tikzset{
  output/.style={color=blue,thick}, 
  nonMaximaPoint/.style={fill=black,thick}, 
  circle/.style={fill}, 
  maxima/.style={circle, color=red},
  upperHull/.style={circle, color=red},
  nonmaxima/.style={circle, color=gray},
  partition/.style={thick, rounded corners, dashed},
  verticalPartition/.style={color=cyan, partition},
  optPartition/.style={color=orange, partition},
  KSAlgorithmDefinedPartition/.style={fill=cyan, thick, rounded corners, dashed},
  ourAlgorithmDefinedPartition/.style={fill=pink, thick, rounded corners, dashed}
}
\newcommand*{\DefineInstanceConvexHullTwoD}{
  \coordinate (m1) at (1,3); 
  \coordinate (m2) at (4,4); 
  \coordinate (m3) at (6,3.5); 
  \coordinate (g1) at (1,0);
  \coordinate (g2) at (3,0);
  \coordinate (g3) at (5,0);
}

\newcommand*{\DefineInstanceMaximaTwoD}{
  \coordinate (m1) at (2,4); 
  \coordinate (m2) at (4,3); 
  \coordinate (m3) at (6,2); 
  \coordinate (m0) at (0,4.2); 
  \coordinate (m4) at (6.2,0);   
  \coordinate (g1) at (0,2); 
  \coordinate (g2) at (2,3); 
  \coordinate (g3) at (4,1); 
  \path
  (g1)
  ++(.5,-.25) node(p1) {}
  ++(.5,-.25) node(p2) {} 
  ++(.5,-.25) node(p3) {} 
  (g2)
  ++(.5,-.25) node(p4) {} 
  ++(.5,-.25) node(p5) {} node(median) {}
  ++(.5,-.25) node(p6) {} 
  (g3)
  ++(.5,-.25) node(p7) {} 
  ++(.5,-.25) node(p8) {} 
  ++(.5,-.25) node(p9) {}
  ;
}
\newcommand*{\DrawInstanceMaximaTwoD}{
  \draw [maxima] (m1) circle (2pt) (m2) circle (2pt) (m3) circle (2pt) ;
  \draw [above right] (m1) node {$q_1$} (m2) node  {$q_2$} (m3) node  {$q_3$} ;
    
  \draw [nonmaxima] 
  (g1)
  ++(.5,-.25) circle (2pt) 
  ++(.5,-.25) circle (2pt) 
  ++(.5,-.25) circle (2pt) 
  (g2)
  ++(.5,-.25) circle (2pt) 
  ++(.5,-.25) circle (2pt) node(median) {}
  ++(.5,-.25) circle (2pt) 
  (g3)
  ++(.5,-.25) circle (2pt) 
  ++(.5,-.25) circle (2pt) 
  ++(.5,-.25) circle (2pt)
  ;
}

\title{Instance-Optimal Geometric Algorithms}
\author{
PEYMAN AFSHANI
\affil{MADALGO, University of Aarhus}
{J\'ER\'EMY} BARBAY
\affil{DCC, Universidad de Chile}
TIMOTHY M. CHAN
\affil{University of Waterloo} 
}

\begin{document}
\maketitle

\begin{abstract}
We prove the existence of an algorithm $A$ for computing 2-d or 3-d convex hulls that is optimal for {\em every point set\/} in the following sense:
for every sequence $\sigma$ of $n$ points and for every algorithm $A'$ in a certain class $\A$, the running time of $A$ on input $\sigma$ is at most a constant factor times the maximum running time of $A'$ on the worst possible permutation of $\sigma$ for $A'$.
In fact, we can establish a stronger property: for every sequence $\sigma$ of points and every algorithm $A'$, the running time of $A$ on $\sigma$ is at most a constant factor times the average running time of $A'$ over all permutations of $\sigma$.
We call algorithms satisfying these properties {\em instance-optimal\/} in the {\em order-oblivious\/} and {\em random-order\/} setting.
Such instance-optimal algorithms simultaneously subsume output-sensitive algorithms and distribution-dependent average-case algorithms, and all algorithms that do not take advantage of the order of the input or that assume the input is given in a random order.

The class $\A$ under consideration consists of all algorithms in a decision tree model where the tests involve only {\em multilinear\/} functions with a constant number of arguments.
To establish an instance-specific lower bound, we deviate from traditional Ben-Or-style proofs and adopt a new adversary argument.
For 2-d convex hulls, we prove that a version of the well known algorithm by Kirkpatrick and Seidel (1986) or Chan, Snoeyink, and Yap (1995) already attains this lower bound.  For 3-d convex hulls, we propose a new algorithm.

We further obtain instance-optimal results for a few other standard problems in computational geometry, such as maxima in 2-d and 3-d, orthogonal line segment intersection in 2-d, finding bichromatic $L_\infty$-close pairs in 2-d, off-line orthogonal range searching in 2-d, off-line dominance reporting in 2-d and 3-d, off-line halfspace range reporting in 2-d and 3-d, and off-line point location in 2-d.  Our framework also reveals connection
to distribution-sensitive data structures and yields  new results as
a byproduct, for example, on on-line orthogonal range searching in 2-d and on-line halfspace range reporting in 2-d and 3-d.
\end{abstract}

\def\citeN{\cite}
\def\shortcite{\cite}
 \section{Introduction}
\label{sec:introduction}

\paragraph{Instance optimality: our model(s)}
Standard worst-case analysis of algorithms has often been criticized as overly pessimistic.  As a remedy, some researchers have turned towards {\em adaptive\/} analysis where the execution cost of algorithms is measured as a function of not just the input size but also other parameters that capture in some ways the difficulty of the input instance.  For example, for problems in computational geometry (the primary domain of the present paper), parameters that have been considered in the past include the output size (leading to so-called {\em output-sensitive\/} algorithms)~\cite{KirkpatrickSeidelSICOMP86}, the spread of an input point set (the ratio of the maximum to the minimum pairwise distance)~\cite{EricksonDCG05}, various measures of fatness of the input objects (e.g., ratio of circumradii to inradii)~\cite{MatousekPachSICOMP94} or clutteredness of a collection of objects~\cite{deBergvanStappenALGMCA02}, the number of reflex angles in an input polygon, and so on.

The ultimate in adaptive algorithms is an {\em instance-optimal\/} algorithm, i.e., an algorithm $A$ whose cost is at most a constant factor from the cost of any other algorithm $A'$ running on the same input, for {\em every\/} input instance.  Unfortunately, for many problems, this requirement is too stringent.
For example, consider the 2-d convex hull problem, which has $\Theta(n\log n)$ worst-case complexity in the algebraic computation tree model: for every input sequence of $n$ points, one can easily design an algorithm $A'$ (with its code depending on the input sequence) that runs in $O(n)$ time on that particular sequence, thus ruling out the existence of an instance-optimal algorithm.\footnote{ The length of the program for $A'$ may depend on $n$ in this example.  If we relax the definition to permit the ``constant factor'' to grow as a function of the program length of $A'$, then an instance-optimal algorithm $A$ exists for many problems such as sorting (or more generally problems that admit linear-time verification).  This follows from a trick attributed to Levin~\cite{JonesBOOK}, of enumerating and simulating all programs in parallel under an appropriate schedule.  To say that algorithms obtained this way are impractical, however, would be an understatement.  }

To get a more useful definition, we suggest a variant of instance optimality where we ignore the order in which the input elements are given, as formalized precisely below:

\newcommand{\OPT}{\mbox{\rm OPT}}
\newcommand{\avg}{^{\mbox{\scriptsize\rm avg}}}
\begin{definition}\label{def:inputOrderObliviousInstanceOptimality}
Consider a problem where the input consists of a sequence of $n$ elements from a domain ${\cal D}$.  Consider a class $\A$ of algorithms.
A {\em correct\/} algorithm refers to an algorithm that outputs a correct answer for every possible sequence of elements in~${\cal D}$.

For a set $S$ of $n$ elements in ${\cal D}$, let $T_A(S)$ denote the maximum running time of $A$ on input $\sigma$ over all $n!$ possible permutations $\sigma$ of $S$.
Let $\OPT(S)$ denote the minimum of $T_{A'}(S)$ over all correct algorithms $A'\in\A$.  If $A\in\A$ is a correct algorithm such that $T_A(S)\le O(1)\cdot\OPT(S)$ for every set $S$, then we say $A$ is {\em instance-optimal in the order-oblivious setting\/}.
\end{definition}

For many problems, the output is a function of the input as a set rather than a sequence, and the above definition is especially meaningful.  In particular, for such problems, instance-optimal algorithms are automatically optimal output-sensitive algorithms; in fact, they are automatically optimal adaptive algorithms with respect to {\em any\/} parameter that is independent of the input order, all at the same time!  This property is satisfied by simple parameters like the spread of an input point set $S$, or more complicated quantities like the expected size $f_r(S)$ of the convex hull of a random sample of size $r$ from $S$~\cite{ClarksonFOCS94}.

For many algorithms (e.g., \texttt{quickhull}~\cite{PreparataShamosBOOK}, to name one), the running time is not affected so much by the order in which the input points are given but by the relative positions of the input points.  Combinatorial and computational geometers more often associate ``bad examples'' with bad point sets rather than bad point sequences.  All this supports the reasonableness and importance of the order-oblivious form of instance optimality.

We can consider a still stronger variant of instance optimality:
\begin{definition}
For a set $S$ of $n$ elements in ${\cal D}$, let $T_A\avg(S)$ denote the average running time of $A$ on input $\sigma$ over all $n!$ possible permutations $\sigma$ of $S$.
Let $\OPT\avg(S)$ denote the minimum of $T_{A'}\avg(S)$ over all correct algorithms $A'\in\A$.
If $A\in\A$ is a correct algorithm such that $T_A(S)\le O(1)\cdot\OPT\avg(S)$ for every set $S$, then we say $A$ is {\em instance-optimal in the random-order setting\/}.\footnote{ One can also consider other variations of the definition, e.g., relaxing the condition to $T_A^{\mbox{\tiny\rm avg}}(S)\le O(1)\cdot\OPT^{\mbox{\tiny\rm avg}}(S)$, or replacing expected running time over random permutations with analogous high-probability statements.  }
\end{definition}

Note that an instance-optimal algorithm in the above sense is immediately also competitive against {\em randomized\/} (Las Vegas) algorithms $A'$, by the easy direction of Yao's principle.  The above definition has extra appeal in computational geometry, as it is common to see the design of randomized algorithms where the input elements are initially permuted in random order~\cite{ClarksonShorDCG89}.

Instance-optimality in the random-order setting also implies {\em average-case\/} optimality where we analyze the expected running time under the assumption that the input elements are random and independently chosen from a common given probability distribution.  (To see this, just observe that the input sequence is equally likely to be any permutation of $S$ conditioned to the event that the set of $n$ input elements equals any fixed set $S$.)  An algorithm that is instance-optimal 
in the random-order setting
can deal with all probability distributions at the same time!
Random-order instance optimality 
also remedies a common complaint about
average-case analysis, that it does not provide information about
an algorithm's performance on a specific input.

\paragraph{Convex hull: our main result}
After making the case for instance-optimal algorithms under our definitions, the question remains: do such algorithms actually exist, or are they ``too good to be true''?  Specifically, we turn to one of the most fundamental and well known problems in computational geometry---computing the convex hull of a set of $n$ points.
Many $O(n\log n)$-time algorithms in 2-d and 3-d have been proposed since the 1970s~\cite{deBergvanKreveldBOOK,EdelsbrunnerBOOK,PreparataShamosBOOK}, which are worst-case optimal under the algebraic computation tree model.
Optimal output-sensitive algorithms can solve the 2-d and 3-d problem in $O(n\log h)$ time, where $h$ is the output size.  The first such output-sensitive algorithm in 2-d was found by \citeN{KirkpatrickSeidelSICOMP86} in the 1980s and was later simplified by \citeN{ChanSnoeyinkYap} and independently \citeN{WengerALGMCA97}; a different, simple, optimal output-sensitive algorithm was discovered by \citeN{ChanDCG96}.
In 3-d, the first optimal output-sensitive algorithm was obtained by \citeN{ClarksonShorDCG89} using randomization; another version was described by \citeN{ClarksonFOCS94}.  The first deterministic optimal output-sensitive algorithm in 3-d was obtained by \citeN{ChazelleMatousekCGTA95} via derandomization; the approach by \citeN{ChanDCG96} can also be extended to 3-d and yields a simpler optimal output-sensitive algorithm.
There are also average-case algorithms running in $O(n)$ expected time for certain probability distributions~\cite{PreparataShamosBOOK}, for example, when the points are independent and uniformly distributed inside a circle or a constant-sized polygon in 2-d, or a ball or a constant-sized polyhedron in 3-d.

The convex hull problem is in some ways an ideal candidate to consider in our models.  It is not difficult to think of examples of ``easy'' point sets and ``hard'' point sets (see Figure~\ref{fig:Instances}(a,b)).
It is not difficult to think of different heuristics for pruning nonextreme points, which may not necessarily improve worst-case complexity but may help for many point sets encountered ``in practice'' (e.g., consider \texttt{quickhull}~\cite{PreparataShamosBOOK}).  However, it is unclear whether there is a single pruning strategy that works best on all point sets.

In this paper, we show that there are indeed instance-optimal algorithms for both the 2-d and 3-d convex hull problem, in the order-oblivious or the stronger random-order setting. Our algorithms thus subsume all the previous output-sensitive and average-case algorithms simultaneously, and are provably at least as good asymptotically as any other algorithm for every point set, so long as input order is ignored.

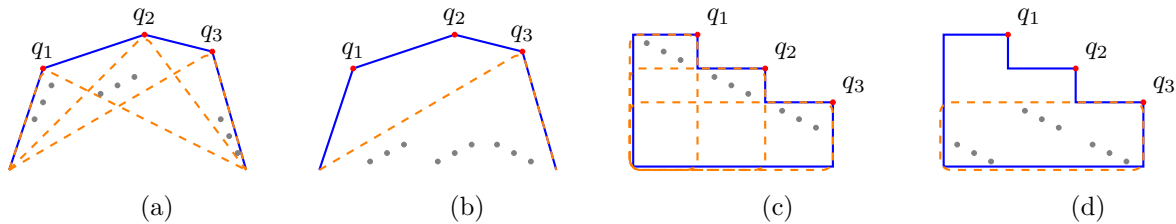
\begin{figure}
\hfill
\begin{minipage}[b]{.24\linewidth} 
\begin{tikzpicture}[scale=0.45] 
  \DefineInstanceConvexHullTwoD
  \draw [output] (0,0) -- (m1)  -- (m2) -- (m3) -- (7,0) ; 
  \draw [maxima] (m1) circle (2pt) (m2) circle (2pt) (m3) circle (2pt) ;
  \draw [above] (m1) node {$q_1$} (m2) node  {$q_2$} (m3) node  {$q_3$} ;
  \draw [optPartition]  (0,0) -- (m1) -- (7,0)  (0,0) -- (m2) -- (7,0)   (0,0) -- (m3) -- (7,0) ;
  \draw [nonmaxima] 
  (.5,1)
  ++(.25,.5) circle (2pt) 
  ++(.25,.5) circle (2pt) 
  ++(.25,.5) circle (2pt) 
  (2.2,2)
  ++(.5,.25) circle (2pt) 
  ++(.5,.25) circle (2pt) node(median) {}
  ++(.5,.25) circle (2pt) 
  (7,0)
  ++(-.25,.5) circle (2pt) 
  ++(-.25,.5) circle (2pt) 
  ++(-.25,.5) circle (2pt)
  ;
\end{tikzpicture}
\center{(a)}
\end{minipage}
\hfill
\begin{minipage}[b]{.24\linewidth} 
\begin{tikzpicture}[scale=0.45]
  \DefineInstanceConvexHullTwoD
  \draw [output] (0,0) -- (m1)  -- (m2) -- (m3) -- (7,0) ; 
  \draw [maxima] (m1) circle (2pt) (m2) circle (2pt) (m3) circle (2pt) ;
  \draw [above] (m1) node {$q_1$} (m2) node  {$q_2$} (m3) node  {$q_3$} ;
  \draw [optPartition]   (0,0) -- (m3) -- (7,0) ;
  \draw [nonmaxima] 
  (1,0)
  ++(.5,.25) circle (2pt) 
  ++(.5,.25) circle (2pt) 
  ++(.5,.25) circle (2pt) 
  (3,0)
  ++(.5,.25) circle (2pt) 
  ++(.5,.25) circle (2pt) node(median) {}
  ++(.5,.25) circle (2pt) 
  (4.75,1)
  ++(.5,-.25) circle (2pt) 
  ++(.5,-.25) circle (2pt) 
  ++(.5,-.25) circle (2pt)
  ;
\end{tikzpicture}
\center{(b)}
\end{minipage}
\hfill
\begin{minipage}[b]{.24\linewidth} 
\begin{tikzpicture}[scale=0.45]
  \DefineInstanceMaximaTwoD
  \draw [output] (0.1,0.1) |- (m1)  |- (m2) |- (m3) |- (0.1,0.1) ; 
  \draw [maxima] (m1) circle (2pt) (m2) circle (2pt) (m3) circle (2pt) ;
  \draw [above right] (m1) node {$q_1$} (m2) node  {$q_2$} (m3) node  {$q_3$} ;
  \draw [optPartition]  (0,0) rectangle (m1)  (0,0) rectangle (m2) (0,0) rectangle (m3) ;
  \draw [nonmaxima] 
  (0,4)
  ++(.5,-.25) circle (2pt) 
  ++(.5,-.25) circle (2pt) 
  ++(.5,-.25) circle (2pt) 
  (2,3)
  ++(.5,-.25) circle (2pt) 
  ++(.5,-.25) circle (2pt) node(median) {}
  ++(.5,-.25) circle (2pt) 
  (4,2)
  ++(.5,-.25) circle (2pt) 
  ++(.5,-.25) circle (2pt) 
  ++(.5,-.25) circle (2pt)
  ;
\end{tikzpicture}
\center{(c)}
\end{minipage}
\hfill
\begin{minipage}[b]{.24\linewidth} 
\begin{tikzpicture}[scale=0.45]
  \DefineInstanceMaximaTwoD
  \draw [output] (0.1,0.1) |- (m1)  |- (m2) |- (m3) |- (0.1,0.1) ; 
  \draw [maxima] (m1) circle (2pt) (m2) circle (2pt) (m3) circle (2pt) ;
  \draw [above right] (m1) node {$q_1$} (m2) node  {$q_2$} (m3) node  {$q_3$} ;
  \draw [optPartition]  (0,0) rectangle (m3); 
  \draw [nonmaxima] 
  (0,1)
  ++(.5,-.25) circle (2pt) 
  ++(.5,-.25) circle (2pt) 
  ++(.5,-.25) circle (2pt) 
  (2,2)
  ++(.5,-.25) circle (2pt) 
  ++(.5,-.25) circle (2pt) node(median) {}
  ++(.5,-.25) circle (2pt) 
  (4,1)
  ++(.5,-.25) circle (2pt) 
  ++(.5,-.25) circle (2pt) 
  ++(.5,-.25) circle (2pt)
  ;
\end{tikzpicture}
   \center{(d)}  
  \end{minipage}
\hfill
  \caption{
%
(a) A ``harder'' point set and (b) an ``easier'' point set for
the 2-d upper hull problem.
(c) A ``harder'' point set and (d) an ``easier'' point set
for the 2-d maxima problem.
%
%
}
  \label{fig:Instances}
\end{figure}

\paragraph{Techniques}
We believe that our techniques---for both the upper-bound side (i.e., algorithms) and the lower-bound side (i.e., proofs of their instance optimality)---are as interesting as our results.

On the upper-bound side, we find that in the 2-d case, a new algorithm is not necessary: a version of Kirkpatrick and Seidel's output-sensitive algorithm~\shortcite{KirkpatrickSeidelSICOMP86}, or its simplification by \citeN{ChanSnoeyinkYap}, is instance-optimal in the order-oblivious and random-order setting.  We view this as a plus: these algorithms are simple and practical to implement~\cite{BhattacharyaSenJALG97}, and our analysis sheds new light on their theoretical complexity.
In particular, our result immediately implies that a version of Kirkpatrick and Seidel's algorithm runs in $O(n)$ expected time for points uniformly distributed inside a circle or a fixed-size polygon---we were unaware of this fact before.
(As another plus, our result provides a positive answer to the
question in the title of Kirkpatrick and Seidel's paper, ``The ultimate
planar convex hull algorithm?'')

In 3-d we propose a new algorithm, as none of the previous output-sensitive algorithms seems to be instance-optimal.
For example, known 3-d generalizations of the Kirkpatrick--Seidel algorithm have suboptimal $O(n\log^2h)$ running time \cite{ChanSnoeyinkYap,EdelsbrunnerShi}, while a straightforward implementation of the algorithm by \citeN{ChanDCG96} fails to be instance-optimal even in 2-d.  Our algorithm builds on Chan's technique~\shortcite{ChanDCG96} but requires additional ideas, notably the use of {\em partition trees\/}~\cite{MatousekDCG92}.

The lower-bound side requires more innovation.  
We are aware of three existing techniques for proving worst-case $\Omega(n\log n)$ (or output-sensitive $\Omega(n\log h)$) lower bounds in computational geometry: (i) information-theoretic or counting arguments, (ii) topological arguments, from early work by \citeN{YaoJACM81} to \citeN{BenOrSTOC83}, and (iii) Ramsey-theory-based arguments, by \citeN{MoranSnirManber}.
Ben-Or's approach is perhaps the most powerful and works in the general algebraic computation tree model, whereas Moran \etal's approach works in a decision tree model where all the test functions have a bounded number of arguments.
For an arbitrary input set $S$ for the convex hull problem, the naive information-theoretic argument gives only an $\Omega(h\log h)$ lower bound on $\OPT(S)$.
On the other hand, topological and Ramsey-theory approaches seem unable to give any instance-specific lower bound (for example, modifying the topological approach is already nontrivial if we just want a lower bound for {\em some\/} integer input set~\cite{YaoSICOMP91}, let alone for {\em every\/} input set, whereas the Ramsey-theory approach considers only input elements that come from a cleverly designed subdomain).

We end up using a different lower bound technique which is inspired by an adversary argument originally used to prove time--space lower bounds for median finding~\cite{ChanSODA09}.
Note that this approach can lead to another proof of the standard $\Omega(n\log n)$ lower bounds for many geometric problems including the problem of computing a convex hull; the proof is simple and works in any algebraic decision tree model where the test functions have at most constant degree and have at most a constant number of arguments.
We build on the idea further and obtain an optimal lower bound for the convex hull problem for {\em every\/} input point set.
The assumed model is more restrictive: the class $\A$ of allowed algorithms consists of those under a decision tree model in which the test functions are {\em multilinear\/} and have at most a constant number of arguments.
Fortunately, most standard primitive operations encountered in existing convex hull algorithms satisfy the multilinearity condition (for example, the standard determinant test does).  
The final proof interestingly involves partition trees~\cite{MatousekDCG92}, which are more typically used in algorithms (as in our new 3-d algorithm) rather than in lower-bound proofs.

So, what is $\OPT(S)$, i.e., what parameter asymptotically captures the true difficulty of a point set~$S$ for the convex hull problem?  As it turns out, the bound has a simple expression (to be revealed in Section~\ref{sec:convex-hull}) and shares similarity with {\em entropy\/} bounds found in average-case
or expected-case 
analysis of geometric data structures where query points come from a given probability distribution---these {\em distribution-sensitive\/} results have been the subject of several previous pieces of work \cite{AryaMalamatosTALG07,AryaMalamatosSICOMP07,ColleteDujmovicSODA08,DujmovicHowatSODA09,IaconoCGTA04}.
However, lower bounds for distribution-sensitive data structures cannot be applied to our problem because our problem is off-line (lower bounds for on-line query problems usually assume that the query algorithms fit a ``classification tree'' framework, but an off-line algorithm may compare a query point not only with points from the data set but also with other query points).
Furthermore, although in the off-line setting we can think of the query points as coming from a discrete point probability distribution, this distribution is not known in advance.\footnote{ {\em Self-improving\/} algorithms~\cite{AilonChazelle} also cope with the issue of how to deal with unknown input probability distributions, but are not directly comparable with our results, since in their setting each point can come from a different distribution, so input order matters.  } 
Lastly, distribution-sensitive data structures are usually concerned with improving the query time, but not the total
time that includes preprocessing.

\paragraph{Other results}
The computation of the convex hull is just one problem for which we are able to obtain instance optimality.  We show that our techniques can lead to instance-optimal results for many other standard problems in computational geometry, in the order-oblivious or random-order setting, including:
\begin{enumerate}
\item[(a)] maxima in 2-d and 3-d; 
\item[(b)] reporting/counting intersection between horizontal and vertical line segments in 2-d; 
\item[(c)] reporting/counting pairs of $L_\infty$-distance at most 1 between a red point set and a blue point set in 2-d; 
\item[(d)] off-line orthogonal range reporting/counting in 2-d; 
\item[(e)] off-line dominating reporting in 2-d and 3-d; 
\item[(f)] off-line halfspace range reporting in 2-d and 3-d; and 
\item[(g)] off-line point location in 2-d.
\end{enumerate}

Optimal expected-case, entropy-based data structures for the on-line version of (g) are known \cite{AryaMalamatosSICOMP07,IaconoCGTA04}, but not for (e,f)---for example, \citeN{DujmovicHowatSODA09} only obtained results for 2-d dominance counting, a special case of 2-d orthogonal range counting.
Incidentally, as a consequence of our ideas, we can also get new optimal expected-case data structures for on-line 2-d general orthogonal range counting and 2-d and 3-d halfspace range reporting.

\paragraph{Related work}
\citeN{Fagin} first coined the term ``instance optimality'' (when studying the problem of finding items with the $k$ top aggregate scores in a database in a certain model), although some form of the concept has appeared before.
For example, the well known ``dynamic optimality conjecture'' is about instance optimality concerning algorithms for manipulating binary search trees (see \cite{DemaineHarmonSODA09} for the latest in a series of papers).
\citeN{DemaineLopezOrtizMunro} studied the problem of computing the union or intersection of $k$ sorted sets and gave instance-optimal results for any $k$ for union, and for constant $k$ for intersection, in the comparison model;
\citeN{BarbayChen} extended their result to the computation of convex hull of $k$ convex polygons in 2-d for constant~$k$.
Another work about instance-optimal geometric algorithms is by \citeN{BaranDemaineIJCGA05}, who addressed an approximation problem about computing the distance of a point to a curve under a certain black-box model.  Other than these, there has not been much work on instance optimality in computational geometry, especially concerning the classical problems under conventional models.

The concept of instance optimality resembles competitive analysis of on-line algorithms.  In fact, in the on-line algorithms literature, our order-oblivious setting of instance optimality is related to what \citeN{BoyarFavrholdtTALG07} called the {\em relative worst order ratio\/}, and our random-order setting is related to what \citeN{KenyonSODA96}
called the  {\em random order ratio}.  What makes instance optimality more intriguing is that we are not bounding the objective function of an optimization problem, but rather the running time of an algorithm.

\section{Warm-Up: 2-d Maxima}
\label{sec:max-points-two}

Before proving our main result on convex hull, we find it useful to study a simpler problem: maxima in 2-d.
For two points $p$ and $q$ we say $p$ {\em{dominates}} $q$ if each coordinate of $p$ is greater than that the corresponding coordinate of $q$.  Given a set $S$ of $n$ points in $\R^d$, a point $p$ is {\em maximal\/} if $p\in S$ and $p$ is not dominated by any other point in $S$.
For simplicity, we assume that the input is always nondegenerate throughout the paper (for example,
no two points share the same $x$- or $y$-coordinate).  The maxima problem is to report all maximal points. 

For an alternative formulation, we can define the {\em orthant\/} at a point $p$ to be the region of all points that are dominated by $p$.
In 2-d, the boundary of the union of the orthants at all $p\in S$ forms a {\em staircase\/}, and the maxima problem is equivalent to computing the staircase of $S$.

This problem has a similar history as the convex hull problem: many worst-case $O(n\log n)$-time algorithms are known; 
an output-sensitive algorithm by \citeN{KirkpatrickSeidelSCG85} runs in $O(n\log h)$ time for output size $h$; and average-case algorithms with $O(n)$ expected time have been analyzed for various probability distributions \cite{BentleyClarksonLevineSODA90,ClarksonFOCS94,PreparataShamosBOOK}.  
The problem is simpler than computing the convex hull, in the sense that direct pairwise comparisons are sufficient.  We therefore work with the class $\A$ of algorithms in the {\em comparison model\/} where we can access the input points only through comparisons of the coordinate of an input point with the corresponding coordinate of another input point.  The number of comparisons made by an algorithm yields a lower bound on the running time.

We define a measure $\entropy(S)$ to represent the difficulty of a point set $S$ and prove that the optimal running time $\OPT(S)$ is precisely $\Theta(n(\entropy(S)+1))$ for the 2-d maxima problem in the order-oblivious and random-order setting.

\begin{definition}
Consider a partition $\Pi$ of the input set $S$ into disjoint subsets $S_1, \ldots, S_t$.
We say that $\Pi$ is {\em{respectful}} if each subset $S_k$ is either a singleton or can be enclosed by an axis-aligned box $B_k$ whose interior is completely below the staircase of $S$.
Define the \emph{entropy} $\entropy(\Pi)$
of the partition $\Pi$ to be
$\sum_{k=1}^{t} (|S_k|/n) \log(n/|S_k|)$.
Define the \emph{structural entropy} $\entropy(S)$
of the input set $S$ to be the minimum of $\entropy(\Pi)$ over all respectful partitions $\Pi$ of $S$.
\end{definition}

\begin{figure}
\centering
\begin{adjustbox}{max width=.3\textwidth}
\begin{tikzpicture}
\DefineInstanceMaximaTwoD%
\DrawInstanceMaximaTwoD%
\draw[optPartition] 
(0.1,0.1) rectangle (4.1,3.1) 
(0.1,3.2) rectangle (2.2,4.1) 
(4.2,0.1) rectangle (6.1,2.1);
\end{tikzpicture}
\end{adjustbox}
\hfill
\begin{adjustbox}{max width=.3\textwidth}
\begin{tikzpicture}
\DefineInstanceMaximaTwoD%
\DrawInstanceMaximaTwoD%
\draw[optPartition] 
(0.2,2.2) rectangle (2.1,4.1) 
(2.2,2.2) rectangle (4.1,3.1) 
(0.2,0.1) rectangle (6.1,2.1) ;
\end{tikzpicture}
\end{adjustbox}
\hfill
\begin{adjustbox}{max width=.3\textwidth}
\begin{tikzpicture}
\DefineInstanceMaximaTwoD%
\DrawInstanceMaximaTwoD%
\draw[verticalPartition] 
(0.25,0.15) rectangle (2.15,4.15) 
(2.25,0.15) rectangle (4.15,3.15) 
(4.25,0.15) rectangle (6.15,2.15) ;
\end{tikzpicture}
\end{adjustbox}
\caption{Three respectful partitions of an instance of the
2-d maxima problem. 
The two partitions on the left have entropy
$\frac{1}{12}\log 12 + \frac{7}{12}\log\frac{12}{7} + \frac{4}{12}\log\frac{12}{4} \approx 1.281$.  
The partition $\PiVert$ on the right
has higher entropy $\log 3\approx 1.585$.
}
\label{fig:partition2dMaxima}
\end{figure}
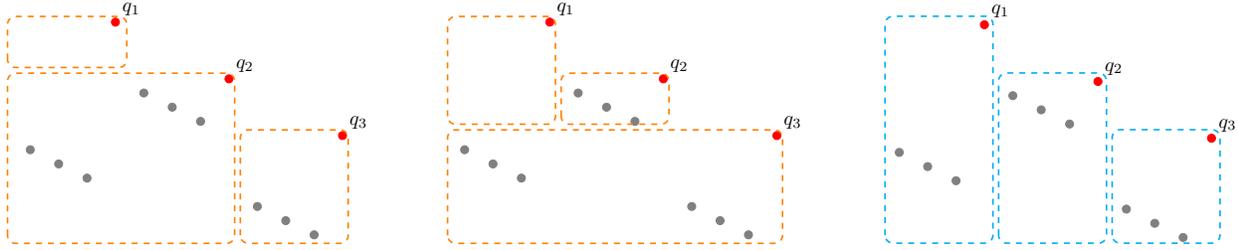

\begin{remark}
Alternatively, we could further insist in the definition that the bounding boxes $B_i$ are nonoverlapping and cover precisely the staircase of $S$.  However, this will not matter, as it turns out that the two definitions yield asymptotically the same quantity (this nonobvious fact is a byproduct of our lower bound proof in Section~\ref{sec:maxima:lb}).

Entropy-like expressions similar to $\entropy(\Pi)$
have appeared in the analysis of
 expected-case geometric data structures for the case of a discrete point probability distribution, although our definition itself is nonprobabilistic.  A measure proposed by \citeN{SenGupta} is identical to $\entropy(\PiVert)$ where $\PiVert$ is a partition of $S$ obtained by dividing the point set $S$ by $h$ vertical lines at the $h$ maximal points of $S$ (see Figure~\ref{fig:partition2dMaxima}(right) for an illustration).  
Note that $\entropy(\PiVert)$ is at most $\log h$ (see Figure~\ref{fig:Instances}(c)) but can be much smaller; in turn, $\entropy(S)$ can be much smaller than $\entropy(\PiVert)$ (see Figures~\ref{fig:Instances}(d) and~\ref{fig:partition2dMaxima}).  The complexity of the 1-d multiset sorting problem~\cite{MunroSpira} also involves an entropy
expression associated with one partition, but does not require taking 
the minimum over multiple partitions. 
\end{remark}

\subsection{Upper bound}
\label{sec:maxima:ub}

We use a slight variant of Kirkpatrick and Seidel's output-sensitive maxima algorithm~\shortcite{KirkpatrickSeidelSCG85} (in their original algorithm, only points from $Q_\ell$ are pruned in line~4):
\begin{algm}
\< \texttt{maxima2d}($Q$):\\
1.\' if $|Q|=1$ then return $Q$\\
2.\' divide $Q$ into the left and right halves $Q_\ell$ and $Q_r$ by the median $x$-coordinate\\
3.\' {\em discover\/} the point $q$ with the maximum $y$-coordinate in $Q_r$\\
4.\' {\em prune\/} all points in $Q_\ell$ and $Q_r$ that are dominated by $q$\\
5.\' return the concatenation of \texttt{maxima2d}$(Q_\ell)$ and \texttt{maxima2d}$(Q_r)$
\end{algm}

We call \texttt{maxima2d}($S$) to start: Figure~\ref{fig:maximaUpperBoundPartialExecution} illustrates the state of the algorithm after a single recursion level.
Kirkpatrick and Seidel showed that its running time is within $O(n\log h)$, and \citeN{SenGupta} improved this upper bound to $O(n(\entropy(\PiVert)+1))$.
Improving this bound to $O(n(\entropy(\Pi)+1))$ for an {\em arbitrary\/} respectful partition $\Pi$ of $S$ requires a bit more finesse:

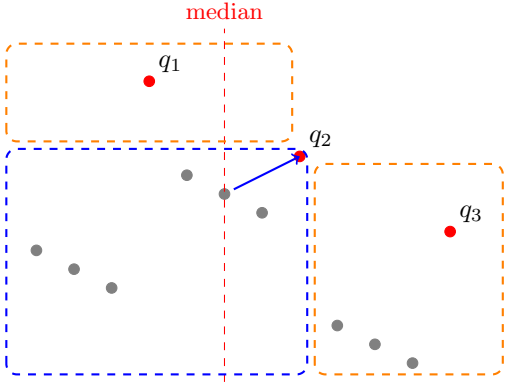
\begin{figure}\centering
\begin{minipage}{.5\textwidth}
\begin{tikzpicture}
\DefineInstanceMaximaTwoD%
\DrawInstanceMaximaTwoD%
\draw[color=red, dashed] (3,0) -- (3,4.7) node[above] {\small median}; 
\draw[->, color=blue, thick] (median) -- (m2); 
\draw[color=blue, partition] (0.1,0.1) rectangle (4.1,3.1); 
\draw[color=orange, partition] (0.1,3.2) rectangle (3.9,4.5); 
\draw[color=orange, partition] (4.2,0.1) rectangle (6.7,2.9);
\end{tikzpicture}
\end{minipage}\begin{minipage}{.4\textwidth}
\caption{Partial execution of \texttt{maxima}$(S)$ after one recursion level. 
In this example, after computing the median $x$-coordinate, the algorithm found the highest point $q_2$ to the right of the median, and pruned the $6$ points dominated by it.  Only $5$ points are left to recurse upon, $1$ to the left and $4$ to the right. \label{fig:maximaUpperBoundPartialExecution} }
\end{minipage}
\end{figure}
\begin{theorem}\label{thm:maxima:ub}
Algorithm \mbox{\tt maxima2d}$(S)$ runs in $O(n(\entropy(S)+1))$ 
time.
\end{theorem}
\begin{proof}
Consider the recursion tree of the algorithm and let $X_j$ denote the sublist of all maximal points of $S$ discovered during the first $j$ recursion levels, in left-to-right order.  Let $S^{(j)}$ be the subset of points of $S$ that {\em survive\/} recursion level $j$, i.e., that have not been pruned during levels $0\TO j$ of the recursion, and let $n_j = |S^{(j)}|$.
The running time is asymptotically bounded by $\sum_{j=0}^{\up{\log n}} n_j$.
Observe that 
\begin{enumerate}
\item[(i)]~there can be at most $\lceil n/2^j\rceil$ points of $S^{(j)}$ with $x$-coordinates between any two consecutive points in $X_j$, and 
\item[(ii)]~all points of $S$ that are strictly below the staircase of $X_j$ have been pruned during levels $0\TO j$ of the recursion.
\end{enumerate}

Let $\Pi$ be {\em any\/} respectful partition of $S$.  Consider a subset $S_k$ in $\Pi$.  Let $B_k$ be a box enclosing $S_k$ whose interior lies below the staircase of $S$.  Fix a level $j$.
Suppose that the upper-right corner of $B_k$ has $x$-coordinate between two consecutive points $q_i$ and $q_{i+1}$ in $X_j$.  By (ii), the only points in $B_k$ that can survive level $j$ must have $x$-coordinates between $q_i$ and $q_{i+1}$.  Thus, by (i), the number of points in $S_k$ that survive level $j$ is at most $\min\left\{|S_k|,\lceil n/2^j\rceil \right\}$.
Since the $S_k$'s cover the entire point set, with a double summation we have
\begin{eqnarray*}
\sum_{j=0}^{\up{\log n}} n_j
&\le&
\sum_{j=0}^{\up{\log n}}\sum_k  \min\left\{|S_k|,\lceil n/2^j \rceil\right\} \\
&=&
\sum_k \sum_{j=0}^{\up{\log n}} \min\left\{|S_k|,\lceil n/2^j \rceil\right\} \\
&\le& 
\sum_k (|S_k|\lceil\log(n/|S_k|)\rceil + |S_k| + |S_k|/2 + |S_k|/4 +\cdots + 1)\\ 
&\le& 
\sum_k |S_k| (\lceil\log(n/|S_k|)\rceil + 2) \\
&\in&
O(n(\entropy(\Pi)+1)). 
\end{eqnarray*}
As $\Pi$ can be \emph{any} respectful partition of $S$, it can be in particular the one of minimum entropy, hence the final result.
\end{proof}

\subsection{Lower bound}
\label{sec:maxima:lb}

For the lower-bound side, we first provide an intuitive justification for the bound $\Omega(n(\entropy(S)+1))$ and point out the subtlety in obtaining a rigorous proof.  Intuitively, to certify that we have a correct answer, the algorithm must gather evidence for each point $p$ eliminated why it is not a maximal point, by indicating at least one {\em witness\/} point in $S$ which dominates $p$.
We can define a partition $\Pi$ by placing points with a common witness in the same subset.  It is easy to see that this partition~$\Pi$ is respectful.  The entropy bound $n\entropy(\Pi)$ roughly represents the number of bits required to encode the partition $\Pi$, so in a vague sense, $n\entropy(S)$ represents the length of the shortest ``certificate'' for~$S$.
\begin{LONG}
\begin{figure}
\centering
\begin{tikzpicture}[scale=.85]
\DefineInstanceMaximaTwoD%
\draw [output] (0.1,0.1) |- (m1)  |- (m2) |- (m3) |- (0.1,0.1) ; 
\DrawInstanceMaximaTwoD%
\draw[ultra thick,color=orange,dashed] (m1)
  edge [->] (m2) 
  edge [->] (m3);
\draw[color=blue,dashed] (m1)
 edge [<-] (p1)  
 edge [<-] (p2)   
 edge [<-] (p3)    ;
\draw[color=blue,dashed] (m2)
 edge [<-] (p4)
 edge [<-] (p5)  
 edge [<-] (p6);
\draw[color=blue,dashed] (m3)
 edge [<-] (p7)   
 edge [<-] (p8)  
 edge [<-] (p9)  ;
\end{tikzpicture}
\begin{tikzpicture}[scale=.85]
\DefineInstanceMaximaTwoD%
\draw [output] (0.1,0.1) |- (m1)  |- (m2) |- (m3) |- (0.1,0.1) ; 
\DrawInstanceMaximaTwoD%
\draw[ultra thick,color=orange,dashed] (m1)
  edge [->] (m2) 
  edge [->] (m3);
\draw[color=blue,dashed] (m2)
 edge [<-] (p1)  
 edge [<-] (p2)   
 edge [<-] (p3)    
 edge [<-] (p4)
 edge [<-] (p5)  
 edge [<-] (p6);
\draw[color=blue,dashed] (m3)
 edge [<-] (p7)   
 edge [<-] (p8)  
 edge [<-] (p9)  ;
\end{tikzpicture}
\caption{The difficulty of proving a lower bound by counting certificates: each elimination of a point $p$ must be justified by indicating a witness which dominates $p$. But there can be more than one such certificate for each instance, and some can be encoded in less space than others. Here the same instance is partitioned on the left into 3 subsets of equal size $4$, while it is partitioned on the right into subsets of sizes
$1,7,4$ with lower entropy.}
\label{fig:AlternativeMaximaAnalysis}
\end{figure}
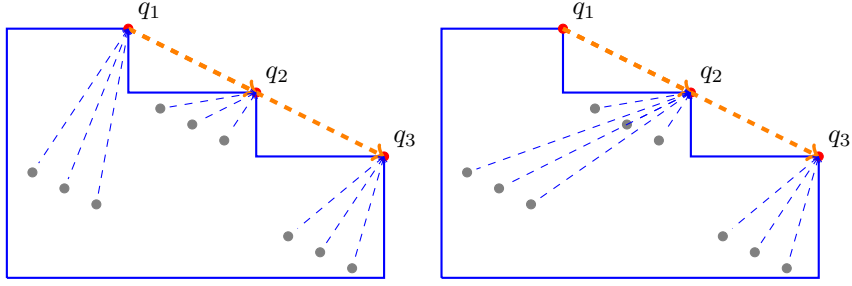
\end{LONG}
Unfortunately, there could be many valid certificates for a given input set $S$ (due to possibly multiple choices of witnesses for each nonmaximal point).
If hypothetically all branches of an algorithm lead to a common partition~$\Pi$, then a straightforward information-theoretic or counting argument would indeed prove the lower bound.  The problem is that each leaf of the decision tree may give rise to a different partition~$\Pi$.
\par
In Appendix~\ref{sec:maxima:alt}, we show that despite the aforementioned difficulty, it is possible to obtain a proof of instance optimality via this approach, but the proof requires a more sophisticated counting argument, and also works with a different difficulty measure.  Moreover, it is limited specifically to the 2-d maxima problem and does not extend to 3-d maxima, let alone to nonorthogonal problems such as the convex hull problem.

In this subsection, we describe a different proof, which generalizes to the other problems that we consider.  The proof is based on an interesting {\em adversary\/} argument.
We show in Section~\ref{sec:randorder} how to adapt the proof to the random-order setting.

\newcommand{\TT}{{\cal T}}
\newcommand{\PiKd}{\Pi_{\mbox{\scriptsize\rm kd-tree}}}
\newcommand{\PiPart}{\Pi_{\mbox{\scriptsize\rm part-tree}}}
\newcommand{\PiPartt}{\widetilde{\Pi}_{\mbox{\scriptsize\rm part-tree}}}

\begin{theorem}\label{thm:maxima:lb}
$\OPT(S)\in\Omega(n(\entropy(S)+1))$ for the 2-d maxima problem in the
comparison model.
\end{theorem}
\begin{proof}
We prove that a specific respectful partition described below not only asymptotically achieves the minimum entropy among all the respectful partitions, but also provides a lower bound for the running time of any comparison-based algorithm that solves the 2-d maxima problem.  The construction of the partition is based
on {\em $k$-d trees\/}~\cite{deBergvanKreveldBOOK}.  We define a tree $\TT$ of axis-aligned boxes, generated top-down as follows:
The root stores the entire plane.
For each node storing box $B$, if $B$ is strictly below the staircase of $S$, or if $B$ contains just one point of $S$, then $B$ is a leaf.
Otherwise, if the node is at an odd (resp.\ even) depth, divide $B$ into two subboxes by the median $x$-coordinate (resp.\ $y$-coordinate) among the points of $S$ inside $B$.
The two subboxes are the children of $B$ (see Figure~\ref{fig:kdTreePartition} for an illustration).
Note that each box $B$ at depth $j$ of $\TT$ contains at least $\lfloor n/2^j\rfloor$ points of~$S$,
and consequently, the depth $j$ is
in $\Omega(\log (n/|S\cap B|))$.
\begin{figure}
\centering
\begin{tikzpicture}
\DefineInstanceMaximaTwoD%
\path[fill=cyan, rounded corners] (0.1,1) rectangle (3.1,4.2); 
  \path[fill=pink, rounded corners] (0.2,1.1) rectangle (1.6,1.9); 
  \draw[color=red, dashed] (0,2) -- (3.1,2) ; 
  \path[fill=pink, rounded corners] (1.7,2.1) rectangle (3.05,4.1); 
\draw[color=red, dashed] (3.15,0) -- (3.15,4.7) node[above] {\small median}; 
\path[fill=cyan, rounded corners] (3.2,0.1) rectangle (6.7,3.3);
  \path[fill=pink, rounded corners] (4.3,0.1) rectangle (5.6,.9); 
  \draw[color=red, dashed] (3.2,1.5) -- (6.7,1.5) ; 
  \path[fill=pink, rounded corners] (3.3,1.6) rectangle (6.2,3.2); 
\draw [output] (0.1,0.1) |- (m1)  |- (m2) |- (m3) |- (0.1,0.1) ; 
\DrawInstanceMaximaTwoD%
\end{tikzpicture}
\caption{The beginning of the recursive partitioning of $S$ by the $k$-d tree $\TT$, which will yield the final partition $\PiKd$ for the adversarial lower bound for the 2-d maxima problem. The two bottom boxes are already leaves, while the two top boxes will be divided further.}
\label{fig:kdTreePartition}
\end{figure}
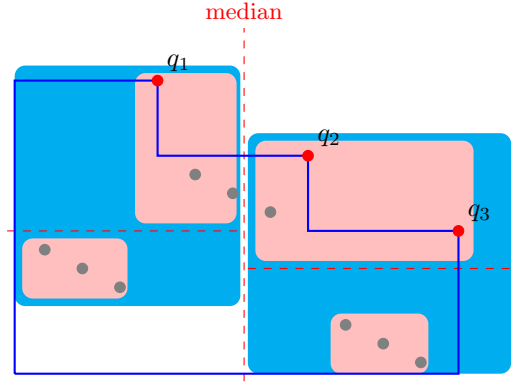

Our claimed partition, denoted by $\PiKd$, is one formed by the leaf boxes in this tree~$\TT$ (i.e., points in the same leaf box are placed in the same subset).
Clearly, $\PiKd$ is respectful.  We will prove that for any correct algorithm $A$ in the comparison model, there exists a permutation of $S$ on which the algorithm requires at least $\Omega(n\entropy(\PiKd))$ comparisons.

The adversary constructs a bad permutation for the input by simulating the algorithm $A$ and resolving each comparison so that the algorithm is forced to perform many others.
During the simulation, we maintain a box $B_p$ in $\TT$ for each point $p$.
If $B_p$ is a leaf box, the algorithm knows the exact location of $p$ inside $B_p$.
But if $B_p$ corresponds to an internal node, the only information the algorithm knows about $p$ is that $p$ lies inside $B_p$.
In other words, $p$ can be assigned any point in $B_p$ without affecting the outcomes of the previous comparisons made.

For each box $B$ in $\TT$, let $n(B)$ be the number of points $p$ such that the box $B_p$ is contained in $B$.
We maintain the invariant that $n(B) \leq |S\cap B|$.
If $n(B) = |S\cap B|$, we say that $B$ is {\em full}.
As soon as $B_p$ becomes a leaf box, we assign $p$ to an arbitrary point in $S\cap B_p$ that has not been previously assigned
(such a point exists because of the invariant); we then call $p$ a {\em fixed\/} point.

Suppose that the algorithm $A$ compares, say, the $x$-coordinates of two points $p$ and $q$.
The main case is when neither $B_p$ nor $B_q$ is a leaf.  
The comparison is resolved in the following way:
\begin{enumerate}
\item If $B_p$ (resp.\ $B_q$) is at even depth, we arbitrarily reset $B_p$ (resp.\ $B_q$) to one of its children that is not full.  Thus assume that $B_p$ and $B_q$ are both at odd depths.
  
  Without loss of generality, suppose that the median $x$-coordinate of $B_p$ is less than the median $x$-coordinate of $B_q$.
  We reset $B_p$ to the left child $B_p'$ of $B_p$ and $B_q$ to the right child $B_q'$ of $B_q$; if either $B_p'$ or $B_q'$ is full, we go to step 2.
  Now, the knowledge that $p$ and $q$ lie in $B_p'$ and $B_q'$ allows us to deduce that $p$ has a smaller $x$-coordinate than $q$.  Thus, the adversary declares to the algorithm that the $x$-coordinate of $p$ is smaller than that of $q$ and continues with the rest of the simulation.
  
\item An exceptional case occurs if $B_p'$ is full (or similarly $B_q'$ is full).  Here, we reset $B_p$ instead to the left (resp. right) sibling $B_p''$ of $B_p'$, but the comparison is not necessarily resolved yet, so we go back to step 1.
\end{enumerate}

Note that in both steps the invariant is maintained.  This is because $B_p$ and $B_p'$ cannot be both full: 
otherwise,
we would have $|S\cap B_p| = |S\cap B_p'| + |S\cap B_p''|=n(B_p')+n(B_p'')$, but $|S\cap B_p|\ge n(B_p)\ge n(B_p')+n(B_p'')+1$ (the ``$+1$'' arises because at 
least one point, notably, $p$, has $B_p$ as its box).



The above description can be easily modified in the
case when $B_p$ or $B_q$ is a leaf box.  If both $B_p$ and $B_q$ are leaf boxes, then $p$ and $q$ are already fixed and the comparison is already resolved.  If (without loss of generality)
only $B_p$ is a leaf, we follow step 1 except that now since $p$ has been fixed, we compare the actual $x$-coordinate of $p$ to the median $x$-coordinate of $B_q$, and reset only~$B_q$.

We now prove a lower bound on the number of comparisons, $T$, made
by the algorithm $A$.
Let $D$ be the sum of the depth of the boxes $B_p$ in the tree $\cal T$ over all points $p\in S$ at the end of the simulation of the algorithm.
We will lower-bound $T$ in terms of $D$.
Each time we reset a box to one of its children in step~1 or~2, $D$ is incremented;
we say that an {\em ordinary\/} (resp.\ {\em exceptional\/}) increment occurs at the parent box if this is done in step~1 (resp.\ step~2).
Each comparison generates only $O(1)$ ordinary increments.
To account for exceptional increments, we use a simple amortization argument: At each box $B$ in $\TT$, the number of ordinary increments has to reach at least $\lfloor |S\cap B|/2\rfloor$ first, before any exceptional increments can occur, and the number of exceptional increments is at most $\lceil |S\cap B|/2\rceil$.
Thus, the total number of exceptional increments can be upper-bounded by the total number of ordinary increments, which is in $O(T)$.
It follows that $D\in O(T)$, i.e., $T\in\Omega(D)$.

Thus, it remains to prove our lower bound for $D$.  
We first argue that at the end of the simulation,
$B_q$ must be a leaf box for every point $q\in S$.
Suppose that this is not the case.
After the end of the simulation, we can do the following postprocessing: for every point $p$ where $B_p$
corresponds to an internal node, we reset $B_p$ to one of its nonfull children arbitrarily, and repeat.
As a result, every $B_p$ now becomes a leaf box, all the input points have been assigned to points of $S$, and no two input points are assigned the same value, i.e., the input is fixed to a permutation of $S$.  The staircase of this input obviously coincides with the staircase of $S$.
Next, consider modifying this input slightly as follows.
Suppose that $B_q$ was not a leaf box before the postprocessing.
Then this box contained at least two points of $S$ and was not completely underneath the staircase of $S$.  We can either move a nonmaximal point upward or a maximal point downward inside~$B_q$, and obtain a modified input that is consistent with the comparisons made but has a different set of maximal points.  The algorithm would be incorrect on this modified input: a contradiction.

%
 It follows that
 \begin{eqnarray*}
 T
 &\in&\Omega(D) \\
 &\subseteq&\Omega\left(\sum_{\mbox{\rm\scriptsize leaf }B} |S\cap B|\log(n/|S\cap B|)\right) \\
 &\subseteq&\Omega(n\entropy(\PiKd)) \\
 &\subseteq&\Omega(n\entropy(S)).
 \end{eqnarray*}
Combined with the trivial $\Omega(n)$ lower bound, this establishes the 
$\Omega(n(\entropy(S)+1))$ lower bound.
\end{proof}

\begin{remark}\label{rmk:maxima:lb}
The above proof is inspired by an adversary argument described by \citeN{ChanSODA09} for a 1-d problem (the original proof maintains a dyadic interval for each input point, while the new proof maintains a box from a hierarchical subdivision).%
\footnote{
There are also some indirect similarities to an adversary argument
for sorting due to \citeN{KahnKim}, 
as pointed out to us
by Jean Cardinal (personal communication, 2010).
}
  The proof still holds for weaker versions of the problem, for example, reporting just the number of maximal points (or the parity of the number).  The lower-bound proof easily extends to any constant dimension and can be easily modified to allow comparisons of different coordinates of any two points $p=(x_1\TO x_d)$ and $q=(x_1'\TO x_d')$, for example, testing whether $x_i < x_j'$, or even $x_i < x_j'+a$ for any constant $a$.  (For a still wider class of test functions, see the next section.)
\end{remark}

\section{Convex Hull}
\label{sec:convex-hull}

We now turn to our main result on 2-d and 3-d convex hull.
It suffices to consider the problem of computing the upper hull of an input point set $S$ in $\R^d\ (d\in\{2,3\})$, since the lower hull can be computed by running the upper hull algorithm on a reflection of $S$.  (Up to constant factors, the optimal running time for convex hull is equal to the maximum of the optimal running time for upper hull and the optimal running time for lower hull, on every input.)

We work with the class $\A$ of algorithms in a {\em multilinear decision tree\/} model where we can access the input points only through tests of the form $f(p_1\TO p_c)>0$ for a multilinear function $f$, over a constant number of input points $p_1\TO p_c$.  
We recall the following standard definition:

\begin{definition}
A function $f:(\R^d)^c\rightarrow\R$ is {\em multilinear\/} if the restriction of $f$ is a linear function from $\R^d$ to $\R$ when any $c-1$ of the $c$ arguments are fixed.  Equivalently, $f$ is multilinear if $f((x_{11}\TO x_{1d})\TO (x_{c1}\TO x_{cd}))$ is a multivariate polynomial function in which each monomial has the form $x_{i_1j_1}\cdots x_{i_kj_k}$ where $i_1\TO i_k$ are all distinct (i.e., we cannot multiply coordinates from the same point).
\end{definition}

Most of the 2-d and 3-d convex hull algorithms that we know fit in this framework: it supports the standard determinant test (for deciding whether $p_1$ is above the line through $p_2,p_3$, or the plane through $p_2,p_3,p_4$), since the determinant is a multilinear function.  
For another example, in 2-d, comparison of the slope of the line through $p_1,p_2$ with the slope of the line through $p_3,p_4$ 
reduces to testing the sign of the function $(y_2-y_1)(x_4-x_3)-(x_2-x_1)(y_4-y_3)$, which is clearly multilinear.
We discuss in Section~\ref{sec:multilinear} the relevance and limitations of the multilinear model.

We adopt the following modified definition of $\entropy(S)$ (as before, it does not matter whether we insist that the simplices $\D_k$ below are nonoverlapping
for both the 2-d and 3-d problem):

\begin{definition}
A partition $\Pi$ of $S$ is {\em{respectful}} if each subset $S_k$ in $\Pi$ is either a singleton or can be enclosed by a simplex $\D_k$ whose interior is completely below the upper hull of $S$.
Define the \emph{structural entropy} $\entropy(S)$ of $S$ to be the minimum of $\entropy(\Pi) = \sum_k (|S_k|/n) \log(n/|S_k|)$ over all respectful partitions $\Pi$ of $S$.
\end{definition}

\subsection{Lower bound}
\label{sec:ch:lb}

The lower-bound proof for computing the convex hull builds on the corresponding lower-bound proof for computing the maxima from Section~\ref{sec:maxima:lb} but is more involved, because a $k$-d tree construction no longer suffices when addressing nonorthogonal problems.  Instead, we use the known following lemma:

\begin{lemma}\label{lem:part}
For every set $Q$ of $n$ points in $\R^d$ and $1\le r\le n$
for any constant~$d$, we can partition $Q$ into $r$ subsets $Q_1\TO Q_r$ each of size $\Theta(n/r)$ and find $r$ convex polyhedral cells $\gamma_1\TO\gamma_r$ each with $O(\log r)$ (or fewer) facets, such that $Q_i$ is contained in $\gamma_i$, and every hyperplane intersects at most $O(r^{1-\eps})$ cells.  Here, $\eps>0$ is a constant that depends only on $d$.
\end{lemma}

The above follows from the {\em partition theorem\/}
of \citeN{MatousekDCG92}, who obtained the best constant 
$\eps=1/d$; in his construction, the cells $\gamma_i$
are simplices (with $O(1)$ facets)
and may overlap, and subset sizes are indeed
upper- and lower-bounded by $\Theta(n/r)$.
(A version of the partition theorem by
\citeN{ChanPartitionTree} can avoid overlapping cells, but does
not guarantee an $\Omega(n/r)$ lower bound on the subset sizes.)

In 2-d or 3-d, a more elementary alternative construction follows from the 4-sectioning or 8-sectioning theorem~\cite{EdelsbrunnerBOOK,YaoDobkinSICOMP89}: for every $n$-point set $Q$ in $\R^2$, there exist 2 lines that divide the plane into 4 regions each with $n/4$ points; for every $n$-point set $Q$ in $\R^3$, there exist 3 planes that divide space into 8 regions each with $n/8$ points.  Since in $\R^2$ a line can intersect at most 3 of the 4 regions and in $\R^3$ a plane can intersect at most 7 of the 8 regions, a simple recursive application of the theorem yields $\eps = 1-\log_4 3$ for $d=2$ and $\eps = 1-\log_8 7$ for $d=3$.  Each resulting cell $\gamma_i$ is 
a convex polytope with $O(\log r)$ facets, and the cells do
not overlap.

We also need another fact, a geometric property about multilinear functions:

\begin{lemma}\label{lem:zero}
If $f:(\R^d)^c\rightarrow \R$ is multilinear and has a zero in $\gamma_1\times\cdots\times\gamma_c$ where each $\gamma_i$ is a convex polytope in $\R^d$, then $f$ has a zero $(p_1\TO p_c)\in \gamma_1\times\cdots\times \gamma_c$ such that all but at most one point $p_i$ is a polytope's vertex.
\end{lemma}
\begin{proof}
Let $(p_1\TO p_c)\in \gamma_1\times\cdots\times \gamma_c$ be a zero of $f$.  Suppose that some $p_i$ does not lie on an edge of $\gamma_i$.  If we fix the other $c-1$ points, the equation $f=0$ becomes a hyperplane, which intersects $\gamma_i$ and thus must intersect an edge of $\gamma_i$.  We can move $p_i$ to such an intersection point.  Repeating this process, we may assume that every $p_i$ lies on an edge $\overline{u_iv_i}$ of $\gamma_i$.  Represent the line segment parametrically as $\{(1-t_i)u_i + t_i v_i\mid 0\le t_i\le 1\}$.

Next, suppose that some two points $p_i$ and $p_j$ are not vertices.
If we fix the other $c-2$ points and restrict $p_i$ and $p_j$ to lie on $\overline{u_iv_i}$ and $\overline{u_jv_j}$ respectively, the equation $f=0$ becomes a multilinear function in two parameters $t_i,t_j\in [0,1]$.
The equation has the form $at_it_j+a't_i+a''t_j+a'''=0$ and is a hyperbola, which intersects $[0,1]^2$ and must thus intersect the boundary of $[0,1]^2$.
We can move $p_i$ and $p_j$ to correspond to such a boundary intersection point.  
Then one of $p_i$ and $p_j$ is now a vertex.
Repeating this process, we obtain the lemma.
\end{proof}

We are now ready for the main lower-bound proof:

\begin{theorem}\label{thm:ch:lb}
$\OPT(S)\in\Omega(n(\entropy(S)+1))$ for the upper hull problem
in any constant dimension $d$ in the multilinear decision tree model.
\end{theorem}
\begin{proof}
We define a {\em partition tree\/} $\TT$ as follows: Each node $v$ stores a pair $(Q(v),\Gamma(v))$, where $Q(v)$ is a subset of $S$ enclosed inside a convex polyhedral cell $\Gamma(v)$.
%
%
The root stores $(S,\R^d)$.
If $\Gamma(v)$ is strictly below the upper hull of $S$, or if $|Q(v)|$ drops below a constant, then $v$ is a leaf.
Otherwise, apply Lemma~\ref{lem:part} with $r=b$ to partition $Q(v)$ and
obtain subsets $Q_1\TO Q_b$ and cells $\gamma_1\TO\gamma_b$.
For the children $v_1\TO v_b$ of $v$, set $Q(v_i)=Q_i$
and $\Gamma(v_i)=\gamma_i\cap\Gamma(v)$.
For a node $v$ at depth $j$ of the tree $\TT$ we then have $|Q(v)| \ge n/\Theta(b)^j$, and consequently,
the depth $j$ is in $\Omega(\log_b (n/|Q(v)|))$.
Furthermore, since $\Gamma(v)$ is the intersection of at most $j$ convex polyhedra with at most $O(\log b)$ facets each, it has size $(j\log b)^{O(1)}$.

Let $\PiPart$ be the partition formed by the subsets $Q(v)$ at the leaves $v$ in $\TT$.
Let $\PiPartt$ be a refinement of this partition obtained as follows: for each leaf $v$ at depth $j$, we triangulate $\Gamma(v)$ into $(j\log b)^{O(1)}$ simplices and subpartition $Q(v)$ by placing points of $Q(v)$ from the same simplex in the same subset; if $|Q(v)|$ drops below a constant, we subpartition $Q(v)$ into singletons.
Note that the subpartitioning of $Q(v)$ causes the entropy to increase%
\footnote{If $\sum_{i=1}^k q_i=q$, then
by concavity of the logarithm,
$\sum_{i=1}^k \frac{q_i}{n}\log\frac{n}{q_i}
\,=\, \frac{q}{n} \sum_{i=1}^k \frac{q_i}{q}\log\frac{n}{q_i}
\,\le\, \frac{q}{n} \log \frac{kn}{q} 
\,=\, \frac{q}{n}\log\frac{n}{q} + \frac{q}{n}\log k$.
}
by at most $O((|Q(v)|/n)\log (j\log b))\subseteq O((|Q(v)|/n)\log\log (n/|Q(v)|))$ for any constant $b$.
The total increase in entropy is thus within $O(\entropy(\PiPart))$.  So $\entropy(\PiPartt)\in\Theta(\entropy(\PiPart))$.  Clearly, $\PiPartt$ is respectful.

The adversary constructs a bad permutation for the input points as follows.  During the simulation, we maintain a node $v_p$ in $\TT$ for each point $p$.  If $v_p$ is a leaf, the algorithm knows the exact location of $p$ inside $\Gamma(v_p)$.  But if $v_p$ is an internal node, the only information the algorithm knows about $p$ is that $p$ lies inside $\Gamma(v_p)$.

For each node $v$ in $\TT$, let $n(v)$ be the number of points $p$ with $v_p$ in the subtree rooted at $v$.  We maintain
the invariant that $n(v) \le |Q(v)|$.  If $n(v) = |Q(v)|$, we say that $v$ is {\em full}.  As soon as $v_p$ becomes a leaf, we fix $p$ to an arbitrary unassigned point in $Q(v_p)$ (such a point exists because of the invariant).  

Suppose that the simulation encounters a test ``$f(p_1\TO p_c) > 0$?''.  The main case is when none of the nodes $v_{p_i}$ is a leaf.
\begin{enumerate}
\item Consider a $c$-tuple $(v_{p_1}'\TO v_{p_c}')$ where $v_{p_i}'$ is a child of $v_{p_i}$.
  We say that the tuple is {\em bad\/} if $f$ has a zero in $\gamma(v_{p_1}')\times\cdots\times\gamma(v_{p_c}')$, and {\em good\/} otherwise.
  We prove the existence of a good tuple by upper-bounding the number of bad tuples:
  If we fix all but one point $p_i$, the restriction of $f$ can have a zero in at most $O(b^{1-\eps})$ cells of the form $\gamma(v_{p_i}')$, by Lemma~\ref{lem:part} and the multilinearity of $f$.  There are $O(b^{c-1}\log b)$ choices of $c-1$ vertices of the cells of the form $\gamma(v_{p_1}')\TO \gamma(v_{p_c}')$.
  By Lemma~\ref{lem:zero}, it follows that the number of bad tuples is at most $O((b^{c-1}\log b)\cdot b^{1-\eps}) \subseteq o(b^c)$.
  As the number of tuples is in $\Theta(b^c)$, if $b$ is a sufficiently large constant, then we can guarantee that some tuple $(v_{p_1}'\TO v_{p_c}')$ is good.
  We reset $v_{p_i}$ to $v_{p_i}'$ for each $i=1\TO c$;
if some $v_{p_i}'$ is full, we go to step~2.  Since the tuple is good, the sign of $f$ is determined and the comparison is resolved.

\item In the exceptional case when some $v_{p_i}'$ is full, we reset $v_{p_i}$ instead to an arbitrary nonfull child, and go back to step~1.
\end{enumerate}

The above description can be easily modified in the case when some of the nodes $v_{p_i}$ are leaves, i.e., when some of the points $p_i$ are already fixed (we just have to lower $c$ by the number of fixed points).

Let $T$ be the number of tests made.  
Let $D$ be the sum of the depth of $v_p$ over all points $p\in S$.
%
%
The same amortization argument as in the previous
proof of Theorem~\ref{thm:maxima:lb}
proves that $T \in \Omega(D)$.
By an argument similar to before,
at the end of the simulation, $v_p$
must be a leaf for every $p\in S$. 
It follows that
\begin{eqnarray*}
T
&\in&\Omega(D)\\
&\subseteq&\Omega\left(\sum_{\mbox{\scriptsize leaf }v} |Q(v)| \log(n/|Q(v)|)\right)\\
&\subseteq&\Omega(n\entropy(\PiPart))\\ 
&\subseteq&\Omega(n\entropy(\PiPartt))\\
&\subseteq&\Omega(n\entropy(S)).
\end{eqnarray*}
Combined with the trivial $\Omega(n)$ lower bound, this establishes the theorem.
\end{proof}

The proof extends to weaker versions of the problem, for example,
reporting the number of hull vertices (or its parity).

\subsection{Upper bound in 2-d}
\label{sec:ch2d}

To establish a matching upper bound in 2-d,
we use a version of the output-sensitive
convex hull algorithm by \citeN{KirkpatrickSeidelSICOMP86}
described below, where an extra pruning step is added in line~2.
(This step is not new and has appeared in both {\tt quickhull}
\cite{PreparataShamosBOOK} and
the simplified output-sensitive
algorithm by \citeN{ChanSnoeyinkYap}; see Figure~\ref{fig:VariantOfKirkpatrickAndSeidelConvexHullAlgorithm} for illustration.)

\begin{figure}
\centering
\begin{tikzpicture}
\coordinate (m1) at (0,0);
\coordinate (m2) at (1,2.5);
\coordinate (m3) at (2,3.5);
\coordinate (m4) at (3,4);
\coordinate (m5) at (5,3.5);
\coordinate (m6) at (6,2.5);
\coordinate (m7) at (7,0);
\draw [ourAlgorithmDefinedPartition] (m1) -| (m4) -- (m1) ;
\draw [ourAlgorithmDefinedPartition] (m7) -| (m5) -- (m7) ;
\draw [KSAlgorithmDefinedPartition] (3,0) -- (m4) -- (m5) -- (5,0) -- (3,0)   ;
\draw[circle]  (m1) circle (2pt) -- (m2) circle (2pt) -- (m3) circle (2pt) -- (m4) circle (2pt) -- (m5) circle (2pt) -- (m6) circle (2pt) -- (m7) circle (2pt);
\draw[upperHull] (m1) circle (2pt)  (m4) circle (2pt)  (m5) circle (2pt)  (m7) circle (2pt);
\end{tikzpicture}
\caption{Kirkpatrick and Seidel's upper hull algorithm~\shortcite{KirkpatrickSeidelSICOMP86} with
an added pruning step.
Line~5 prunes the points in the shaded trapezoid.
The added step in line~2 prunes points
in the two shaded triangles.}
\label{fig:VariantOfKirkpatrickAndSeidelConvexHullAlgorithm}
\end{figure}
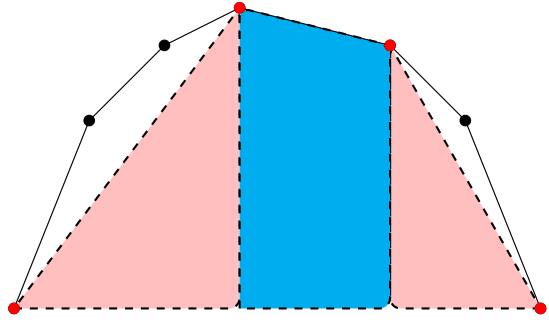

\begin{algm}
\< \texttt{hull2d}($Q$):\\
1.\' if $|Q|=2$ then return $Q$\\
2.\' {\em prune\/} all points from $Q$ strictly below the line through the leftmost and \\  rightmost points of $Q$\\
3.\' divide $Q$ into the left and right halves $Q_\ell$ and $Q_r$
by the median $x$-coordinate $p_m$\\
4.\' {\em discover\/} points $q,q'$ that define
the upper-hull edge $\overline{qq'}$ intersecting
the vertical \\ line at $p_m$ (in linear time)\\
5.\' {\em prune\/} all points from $Q_\ell$ and $Q_r$ that are strictly
underneath the line segment $\overline{qq'}$\\
6.\' return the concatenation of \texttt{hull2d}$(Q_\ell)$ and \texttt{hull2d}$(Q_r)$
\end{algm}

Line~4 can be done in $O(n)$ time (without knowing the upper hull beforehand) by applying a 2-d linear programming algorithm in the dual~\cite{PreparataShamosBOOK}.
We call \texttt{hull2d}($S$) to start.  It is straightforward to show that the algorithm, even without line~2, runs in time $O(n\log h)$, or $O(n(\entropy(\PiVert)+1))$ for the specific partition $\PiVert$ of $S$ obtained by placing points underneath the same upper-hull edge in the same subset, as was done by \citeN{SenGupta}.  To upper-bound the running time by $O(n(\entropy(\Pi)+1))$ for an arbitrary respectful partition $\Pi$ of $S$, we modify the proof in Theorem~\ref{thm:maxima:ub}:

\begin{theorem}
Algorithm \mbox{\tt hull2d}$(S)$ runs in $O(n(\entropy(S)+1))$ time.
\end{theorem}
\begin{proof}
Like before, let $X_j$ denote the sublist of all hull vertices discovered during the first $j$ levels of the recursion, in left-to-right order.  Let $S^{(j)}$ be the subset of points of $S$
that survive recursion level~$j$, and $n_j=|S^{(j)}|$.
The running time is asymptotically bounded by 
$\sum_{j=0}^{\up{\log n}} n_j$.
Observe that 
\begin{enumerate}
\item[(i)]~there can be at most $\lceil n/2^j\rceil$ points of $S^{(j)}$ with $x$-coordinates between any two consecutive vertices in $X_j$, and 
\item[(ii)]~all points that are strictly below the upper hull of $X_j$ have been pruned during levels $0\TO j$ of the recursion.
\end{enumerate}

Let $\Pi$ be any respectful partition of $S$.  Consider a subset $S_k$ in $\Pi$.  Let $\D_k$ be a triangle enclosing $S_k$ whose interior lies below the upper hull of $S$.  Fix a level $j$.
If $q_i$ and $q_{i+1}$ are two consecutive vertices in $X_j$ such that $\overline{q_iq_{i+1}}$ does not intersect the boundary of $\D_k$ (i.e., is above $\D_k$), then all points in $\D_k$ with $x$-coordinates between $q_k$ and $q_{k+1}$ would have been pruned during the first $j$ levels by (ii).
Since only $O(1)$ edges $\overline{q_iq_{i+1}}$ of the upper hull of $X_j$ can intersect the boundary of $\D_k$, the number of points in $S_k$ that survive level $j$ is at most $\min\left\{|S_k|,O(n/2^j)\right\}$ by (i).
We then have
 $$\sum_{j=0}^{\up{\log n}} n_j\ \in\ \sum_{j=0}^{\up{\log n}}\sum_k  \min\left\{|S_k|,O(n/2^j)\right\} \ \subseteq\ O(n(\entropy(\Pi)+1))$$
as before.
\end{proof}

\begin{remark}
The same result holds for the simplified output-sensitive algorithm
by \citeN{ChanSnoeyinkYap}, which avoids the need to invoke a 2-d linear programming algorithm.  (Chan \etal's paper explicitly added the pruning step in their algorithm description.)  The only difference in the above analysis is that there can be at most $\lceil (3/4)^j n\rceil$ points of $S$ with $x$-coordinates between any two consecutive vertices in $X_j$.
\end{remark}

\subsection{Upper bound in 3-d}
\label{sec:ch3d}

We next present an instance-optimal
algorithm in 3-d that matches our lower bound.
Unlike in 2-d, it is unclear if any of the known algorithms can be modified for this purpose.  
For example, obtaining an $O(n(\entropy(\PiVert)+1))$ upper bound is already
nontrivial for the specific partition $\PiVert$ where points underneath the same upper-hull facet are placed in the same subset.
Informed by our lower-bound proof, we suggest an algorithm that is also based on partition trees.
We need the following subroutine:

\begin{lemma}\label{lem:group}
Given a set of $n$ halfspaces in $\R^d$ for any constant $d$,
we can answer a sequence of $r$ linear programming queries 
(finding the point that maximizes a query linear function over
the intersection of the halfspaces)
in total time $O(n\log r + r^{O(1)})$.
\end{lemma}

The above lemma was obtained by Chan~\shortcite{ChanDCG96,ChanDCG96b}
using a simple grouping trick (which was
the basis of his output-sensitive $O(n\log h)$-time convex
hull algorithm); the $d\ge 3$ case required randomization.
A subsequent paper by \citeN{ChanSCG96}
gave an alternative approach using a partition construction; this eliminated randomization.

Our new upper hull algorithm can now be described as follows:
\begin{algm}
\< \texttt{hull3d}($Q$):\\
1.\' for $j=0,1\TO \lfloor\log(\de \log n)\rfloor$ do\\
2.\'\> partition $Q$ by Lemma~\ref{lem:part} to get $r_j=2^{2^j}$
	 subsets   
	$Q_1\TO Q_{r_j}$ and\\\> cells $\gamma_1\TO\gamma_{r_j}$\\
3.\'\> for each $i=1$ to $r_j$ do\\
4.\'\>\> if $\gamma_i$ is strictly below the upper hull of $Q$  then 
{\em prune\/} all points in $Q_i$ from $Q$\\
5.\' compute the upper hull of the remaining set $Q$ directly
\end{algm}

Line~2 takes $O(|Q|\log r_j + r_j^{O(1)})$ time by known algorithms for
Matou\v sek's partition theorem~\shortcite{MatousekDCG92} 
(or alternatively recursive application of the 8-sectioning theorem).
The test in line~4 reduces to deciding whether each
of the at most 
$O(\log r_j)$ vertices of the convex polyhedral cell $\gamma_i$ is strictly below the upper hull of $Q$.
This can be done (without knowing the upper hull beforehand) by answering a 3-d linear programming query in dual space.
Using Lemma~\ref{lem:group}, we can perform lines 3--4 collectively in time $O(|Q|\log r_j + r_j^{O(1)})$.
As $r_j\le n^\delta$,
the $r_j^{O(1)}$ term is negligible, since its total over all iterations is sublinear in $n$ by choosing a small constant $\delta$.
Line~5 is done by running any $O(|Q|\log|Q|)$-time algorithm.

\begin{theorem}\label{thm:ch3d}
Algorithm \mbox{\tt hull3d}$(S)$ runs in $O(n(\entropy(S)+1))$ time.
\end{theorem}
\begin{proof}
Let $n_j$ be the size of $Q$ just after iteration~$j$.  The total running time is asymptotically bounded by $\sum_j n_j\log r_{j+1}$.
(This includes the cost of line~5, which is $O(n_j \log n_j)\in O(n_j \log r_{j+1})$
for the last index $j=\lceil\log(\delta \log n)\rceil$.)

Let $\Pi$ be any respectful partition of $S$.  Consider a subset $S_k$ in $\Pi$.
Let $\D_k$ be a simplex enclosing $S_k$ whose interior lies below the upper hull of $S$.  Fix an iteration $j$.  Consider the subsets $Q_1\TO Q_{r_j}$ and cells $\gamma_1\TO\gamma_{r_j}$ at this iteration.
If a cell $\gamma_i$ is completely inside $\D_k$, then all points inside $\gamma_i$ are pruned.
Since $O(r_j^{1-\eps})$ cells $\gamma_i$ intersect the boundary of $\D_k$, the number of points in $S_k$ that remain in $Q$ after iteration $j$ is at most $\min\left\{|S_k|, O(r_j^{1-\eps}\cdot n/r_j)\right\} = \min\left\{|S_k|, O(n/r_j^\eps) \right\}$.  The $S_k$'s cover the entire point set, so with a double summation we have
\begin{eqnarray*}
\sum_j n_j\log r_{j+1} &\le&
\sum_j \sum_k \min\left\{|S_k|,\: O\left({n \over 2^{\eps 2^j}}\right) \right\}\cdot 2^{j+1}\\
&=& \sum_k \sum_j \min\left\{|S_k|,\: O\left({n \over 2^{\eps 2^j}}\right) \right\}\cdot 2^{j+1}\\
&\in& \sum_k O\left( \sum_{j\le\log((1/\eps)\log(n/|S_k|))+1} |S_k|2^j
                  \ +\sum_{j>\log((1/\eps)\log(n/|S_k|))+1} {n \over 2^{\eps 2^{j-1}}}
           \right)\\
 &\in& \sum_k O\left(|S_k| (\log(n/|S_k|) + 1)\right)\\
 &\in&  O(n(\entropy(\Pi)+1)),
\end{eqnarray*}
which yields the theorem.
\end{proof}

\begin{remark}
Variants of the algorithm are possible.  For example, instead of recomputing the partition in line~3 at each iteration from scratch, another option is to build the partitions hierarchically as a tree.  Points are pruned as the tree is generated level by level.

One minor technicality is that the above description of the algorithm does not discuss the low-level test functions involved.
In Section~\ref{sec:multilinear} we explain how a modification of the algorithm can indeed be implemented in the multilinear model.

A similar approach works for the 3-d maxima problem in the comparison model.  We just replace partition trees with
$k$-d trees, and replace linear programming queries with queries to test whether a point lies underneath the staircase, which can be done via an analog of Lemma~\ref{lem:group}.
\end{remark}

\section{Extension to the Random-Order setting}
\label{sec:randorder}

In this section, we describe how our lower-bound proofs in the order-oblivious setting can be adapted to the random-order setting.  
We focus on the convex hull problem and
describe how to modify the proof
of Theorem~\ref{thm:ch:lb}.
We need a technical lemma first:

\begin{lemma}\label{lem:randorder}
Suppose we place $n$ random elements independently in $t$ bins, where each element is placed in the $k$-th bin with probability $n_k/n$.  Then the probability that the $k$-th bin contains exactly $n_k$ elements for all $k=1\TO t$ is at least $n^{-O(t)}$.
\end{lemma}
\begin{proof}
The probability is exactly
$\frac{n!}{n_1!\cdots n_t!}
\left(\frac{n_1}{n}\right)^{n_1}\cdots
\left(\frac{n_t}{n}\right)^{n_t},$
which by Stirling's formula is 
\[ \frac{\Theta(\sqrt{n})(n/e)^{n}}{
\Theta(\sqrt{n_1})(n_1/e)^{n_1}\cdots \Theta(\sqrt{n_t})(n_t/e)^{n_t}
}
\left(\frac{n_1}{n}\right)^{n_1}\cdots
\left(\frac{n_t}{n}\right)^{n_t}
\ \subseteq\ \frac{1}{\Theta(\sqrt{n})^{t-1}},
\]
yielding the result.
\end{proof}

We now present our lower-bound proof in
the random-order setting.  (The proof
is loosely inspired by the randomized ``bit-revealing'' argument by \citeN{ChanSODA09}.)

\newcommand{\PiKdd}{\widetilde{\Pi}_{\mbox{\scriptsize\rm kd-tree}}}

\begin{theorem}\label{thm:randorder}
$\OPT\avg(S)\in\Omega(n(\entropy(S)+1))$ for the upper hull 
problem in any constant dimension~$d$ in the multilinear decision tree model.
\end{theorem}
\begin{proof}
Fix a sufficiently small constant $\de>0$.  Let $\TT$ be as in the proof of Theorem~\ref{thm:ch:lb}, except that we keep only the first $\down{\de\log n}$ levels of the tree, i.e., when a node reaches depth $\down{\de\log n}$, it is made a leaf.

Let $\PiPart$ be the partition of $S$ formed by the leaf cells in $\TT$.  Let $\PiPartt$ be a refinement of $\PiPart$ in which each
leaf cell is further triangulated  
and each subset corresponding to a cell of depth $\down{\de\log n}$ is further subpartitioned into singletons.  Note that each such subset has size $\Theta(n^\de)$ and contributes $\Theta((n^\de/n)\log n)$ to both the entropy of $\PiPart$ and $\PiPartt$.  Thus, $\entropy(\PiPartt)\in\Theta(\entropy(\PiPart))$.  Clearly, $\PiPartt$ is respectful.

The adversary proceeds differently.  We do not explicitly maintain the invariant that no node $v$ is full.  Whenever some $v_p$ first becomes a leaf, we assign $p$ to a random point among the points in $Q(v_p)$ that has previously not been assigned.  If all points in $Q(v_p)$ have in fact been assigned, we say that {\em failure\/} has occurred.

Suppose that the simulation encounters a 
test ``$f(p_1\TO p_c) > 0$''.  We do the following:
\begin{itemize}
\item We reset each $v_{p_i}$ to one of its children at random, where each child $v_{p_i}'$ is chosen with probability
$|Q(v_{p_i}')|/|Q(v_{p_i})|$ (which is in $\Theta(1/b)$).  If the tuple $(v_{p_1}'\TO v_{p_c}')$ is good (as defined in
the proof of Theorem~\ref{thm:ch:lb}),
then the comparison is resolved.  Otherwise, we repeat.
\end{itemize}

Since we have shown that the number of bad tuples is in $o(b^c)$,
the probability that the test is not resolved in one step
is in $o(b^c)\cdot \Theta(1/b)^c$, which can be made 
less than $1/2$ for a sufficiently large constant $b$.  The number of iterations per comparison is thus upper-bounded by a geometrically distributed random variable with mean $O(1)$.

Let $T$ be the number of comparisons made.  Let $D$ be the sum of the depth of $v_p$ over all points $p\in S$ at the end of the simulation.  Clearly, $D$ is upper-bounded by the total number of iterations performed, which is at most a sum of $T$ independent geometrically distributed random variables with mean $O(1)$.  Let $(\ast)$ be the event that $D\le c_0T$ for a sufficiently large constant $c_0$.  By the Chernoff bound, $\Pr[(\ast)]\ge 1-2^{-\Omega(T)}\ge 1-2^{-\Omega(n)}$.


By the same argument as before, at the end of
the simulation, $v_p$ must be a leaf for every
$p\in S$, assuming that failure has not occurred.

Let $(\dag)$ be the event that failure has not occurred.  If $(\ast)$ and $(\dag)$ are both true, then
\begin{eqnarray*}
T&\in&\Omega(D)\\
&\subseteq& \Omega\left(\sum_{\mbox{\rm\scriptsize leaf }v} |Q(v)|\log(n/|Q(v)|)\right)\\
&\subseteq& \Omega(n\entropy(\PiPart))\\
&\subseteq& \Omega(n\entropy(\PiPartt))\\
&\subseteq& \Omega(n\entropy(S)).
\end{eqnarray*}

To analyze $\Pr[(\dag)]$, consider the leaf $v_p$ that a point $p$ ends up with after the simulation (regardless of whether failure has occurred).  This is a random variable, which equals a fixed leaf $v$ with probability $|Q(v)|/n$.  Moreover, all these random variables are independent.  Failure occurs if and only if for some leaf $v$, the number of $v_p$'s that equal $v$ is different from $|Q(v)|$.  By Lemma~\ref{lem:randorder}, $\Pr[(\dag)] \ge n^{-O(n^\de)}$, since there are $O(n^\de)$ leaves in $\TT$.
It follows that
\[ 
  \Pr[\mbox{not }(\ast)\mid (\dag)]\ \le\ 
   \frac{\Pr[\mbox{not }(\ast)]}{\Pr(\dag)}\ \in\ 
   \frac{2^{-\Omega(n)}}{n^{-O(n^{\de})}}
   \ \subseteq\ 2^{-\Omega(n)}.
\]

Finally, observe that $\Pr[(\dag) \wedge \mbox{(the input equals $\sigma$)}]$ is the same for all fixed permutations $\sigma$ of $S$ (the probability is exactly $\prod_{\mbox{\rm\scriptsize leaf }v} \left(\frac{|Q(v)|}{n}\right)^{|Q(v)|} \frac{1}{|Q(v)|!}$).
In other words, conditioned to $(\dag)$, the input is a random permutation of $S$, i.e., the adversary have not acted adversarily after all!
It follows that $T\in\Omega(n\entropy(S))$ with high probability for a random permutation of $S$.
In particular,
$\Ex[T]\in\Omega(n\entropy(S))$ for a random permutation of~$S$.
\end{proof}

\begin{remark}
Applying the same ideas to the proof of Theorem~\ref{thm:maxima:lb} shows that $\OPT\avg(S)\in\Omega(n(\entropy(S)+1))$ for the maxima problem in the comparison model.
\end{remark}

\section{On the Multilinear Model}
\label{sec:multilinear}

\newcommand{\ABOVE}{\mbox{\sc above}} \newcommand{\PLANE}{\mbox{\sc
    plane}} \newcommand{\INTERSECT}{\mbox{\sc intersect}}
\newcommand{\DEF}{\mbox{\sc def}}

We remark that if nonmultilinear test functions are allowed,
then $n\entropy(S)$ may no longer a valid 
instance-optimal lower bound under our definition of $\entropy(S)$.  For example, one can design both an instance $S$ of the 2-d maxima problem with $h$ output points, 
having $\entropy(S)\in\Omega(\log h)$ (see Figure~\ref{fig:nonMultilinearAdversaryAlgorithn}), and an algorithm $A$ that requires just $O(n + h\log h)$ operations on that instance using nonmultilinear tests.
A similar example can be constructed for the 2-d convex hull problem.  

\begin{figure}
\centering
\begin{tikzpicture}
\draw[fill=cyan] (0,0) -- (0,4.1) arc (90:0:4.1) -- (0,0) ;
\draw[nonMaximaPoint] (0.1,4.05) circle (2pt);
\draw[nonMaximaPoint] (1.5,3.5) circle (2pt);
\draw[nonMaximaPoint] (3.2,2.4) circle (2pt);
\draw[nonMaximaPoint] (3.8,1.2) circle (2pt);
\draw[nonMaximaPoint] (4.05,0.1) circle (2pt);

\draw[nonMaximaPoint] (0.4,4.05) circle (2pt);
\draw[nonMaximaPoint] (1.7,3.2) circle (2pt);
\draw[nonMaximaPoint] (3.4,2.2) circle (2pt);
\draw[nonMaximaPoint] (3.7,1.4) circle (2pt);
\draw[nonMaximaPoint] (4.05,0.3) circle (2pt);
\coordinate (m1) at (1,4.5);
\coordinate (m2) at (3,4);
\coordinate (m3) at (3.6,3);
\coordinate (m4) at (4,2.1);
\coordinate (m5) at (4.5,1);
\draw[maxima] (m1) circle (2pt);
\draw[maxima] (m2) circle (2pt);
\draw[maxima] (m3) circle (2pt);
\draw[maxima] (m4) circle (2pt);
\draw[maxima] (m5) circle (2pt);
\draw[partition] (0,0) |- (m1);
\draw[partition] (0,0) |- (m2);
\draw[partition] (0,0) |- (m3);
\draw[partition] (0,0) |- (m4);
\draw[partition] (0,0) |- (m5);
\draw[partition] (0,0) -| (m1);
\draw[partition] (0,0) -| (m2);
\draw[partition] (0,0) -| (m3);
\draw[partition] (0,0) -| (m4);
\draw[partition] (0,0) -| (m5);

\end{tikzpicture}
\caption{An instance where $n\entropy(S)$ is no longer
a lower bound if nonmultilinear tests are allowed.
  A circular disk covers all nonmaximal points underneath
the staircase.
An algorithm tailored to this instance
can identify the $n-h$ points that are
inside the disk by $O(n)$ nonmultilinear tests, then
compute the staircase of the remaining $h$ points
in $O(h\log h)$ time, and verify that the disk is
underneath the staircase by $O(h)$ additional nonmultilinear tests.
}
\label{fig:nonMultilinearAdversaryAlgorithn}
\end{figure}
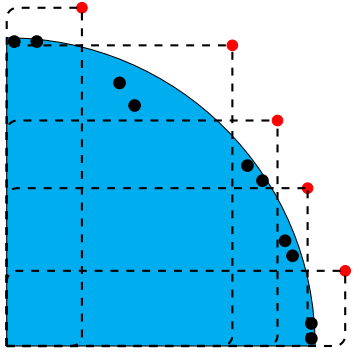

Nevertheless, many standard test functions commonly found in geometric algorithms are multilinear.  For example, in 3-d, the predicate $\ABOVE(p_1\TO p_4)$ which returns true if and only if $p_1$ is above the plane through $p_2,p_3,p_4$ can be reduced to testing the sign of a multilinear function (a determinant).

To see the versatility of multilinear tests,
consider the following extended definition:
we say that a function $f:(\R^d)^c\rightarrow\R^d$ is {\em quasi-multilinear\/} if $f(p_1\TO p_c) = (f_1(p_1\TO p_c)\TO f_d(p_1\TO p_c))$ where $f_i = h_i(p_1\TO p_c)/g(p_1\TO p_c)$ in which $f_1\TO f_d,g:(\R^d)^c\rightarrow\R$ are multilinear functions.   In 3-d, the function $\PLANE(p_1,p_2,p_3)$ which returns the dual of the plane through $p_1,p_2,p_3$ is quasi-multilinear; similarly the function $\INTERSECT(p_1,p_2, p_3)$ which returns the intersection of the dual planes of $p_1,p_2,p_3$ is quasi-multilinear. 
This can be seen by expressing the answer as a ratio of
determinants. 

More generally, we have the following rules:
\begin{itemize}
\item
if the function $f_i:(\R^3)^{c_i}\rightarrow\R^3$ is quasi-multilinear
for each $i\in\{1,2,3\}$, then
$\PLANE(f_1(p_{11}\TO p_{1c_1}),f_2(p_{21}\TO p_{2c_2}), f_3(p_{31}\TO p_{3c_3}))$ and
$\INTERSECT(f_1(p_{11}\TO p_{1c_1}),f_2(p_{21}\TO p_{2c_2}),f_3(p_{31}\TO p_{3c_3}))$  
are quasi-multilinear;
\item
if the function $f_i:(\R^3)^{c_i}\rightarrow\R^3$ is quasi-multilinear
for each $i\in\{1\TO 4\}$, then
$\ABOVE(f_1(p_{11}\TO p_{1c_1})\TO f_4(p_{41}\TO p_{4c_4}))$
can be reduced to testing the sign of a multilinear function.
\end{itemize}


By combining the above rules, more and more elaborate predicates can thus be reduced to testing the signs of
multilinear functions, such as in the following example:
\[
\ABOVE\left(\begin{array}{l}
  p_{10},\\
  p_{11},\\
  p_{12}, \\
  \INTERSECT\left(
  \begin{array}{l}
     \PLANE(p_1,p_2,p_3),\\
     \PLANE(p_4,p_5,p_6),\\
     \PLANE(p_7,p_8,p_9)
     \end{array}\right)
\end{array}\right).
\]
However, we may run into problems if a point occurs more than once
in the expression, as in the following example:
\[
\ABOVE\left(\begin{array}{l}
  p_{10},\\
  p_{11},\\
  p_{1}, \\
  \INTERSECT\left(
  \begin{array}{l}
     \PLANE(p_1,p_2,p_3),\\
     \PLANE(p_4,p_5,p_6),\\
     \PLANE(p_7,p_8,p_9)
     \end{array}\right)
\end{array}\right).
\]
Here, the expansion of the determinants may yield monomials of the wrong type.  In most 2-d algorithms, this kind of tests does not arise or can be trivially eliminated.  
Unfortunately, they can occasionally occur in some 3-d
algorithms, including our 3-d upper hull algorithm in Section~\ref{sec:ch3d}.  We now describe how to modify our
algorithm to avoid these problematic tests.

First, we consider the partition construction in Lemma~\ref{lem:part}.  We choose the more elementary
alternative based on the 8-sectioning theorem: 
there exist 3
planes that divide space into 8 regions, each with $n/8$ points
of $Q$.  By perturbing the 3 planes one by one, we can
ensure that each of the 3 planes passes through 3 input points,
and that the resulting 9 points are distinct, while
changing the number of points of $Q$ in each region by
$\pm O(1)$.
A brute-force algorithm can find 3 such planes in polynomial time.  We can reduce the construction time by using the standard tool called {\em epsilon-approximations}~\cite{MatousekHANDBOOK00}:
we compute a $\delta$-approximation of $Q$ in linear time
for a sufficiently small constant $\delta>0$, and then
apply the polynomial algorithm to the constant-sized $\delta$-approximation.  This only changes the fraction $1/8$ by
a small term $\pm O(\delta)$.  It can be checked that known 
algorithms for epsilon-approximations~\cite{MatousekHANDBOOK00} require only multilinear tests (it suffices to check the implementation of the so-called {\em subsystem oracle}, which only
requires the \ABOVE\ predicate).  We remove the 9 defining
points before recursively proceeding inside the 8 regions.
As a result, we can ensure that the facets in each convex polyhedral cell are all defined by planes that pass through 3 input points, where no two planes share a common defining point.  A vertex $v$ of a cell is an intersection of 3 such planes and is defined by a set of 9 distinct input points, denoted $\DEF(v)$.

Next, we consider the proof of Lemma~\ref{lem:group} for
answering $r$ linear programming queries.
We choose the alternative approach by \citeN{ChanSCG96}
based on a partition construction, which we have from the
previous paragraph.  This algorithm is based on
a deterministic version of the sampling-based linear
programming algorithm by \citeN{ClarksonLP}.  The algorithm
can also support up to $r$ insertions and deletions
of halfspaces intermixed with the query sequence. 
The algorithm can be implemented with simple predicates such as
$\ABOVE$.

Now, in the algorithm \texttt{hull3d}, we make one change:
in line~5, we prune only when each vertex $v$ of $\gamma_i$ lies strictly below the upper hull of $Q-\DEF(v)$ (instead of
the upper hull of $Q$).  In the dual, testing such a vertex $v$
reduces to a linear programming query after deletion of
$\DEF(v)$ from $Q$, where the coefficient
vector of the objective function is quasi-multilinear in
$\DEF(v)$.  Since $\DEF(v)$ has been deleted from $Q$, we avoid the problem
of test functions where some point appears more than once in the
expression.  It can be checked that applying the algorithm
for linear programming queries from the previous paragraph
indeed requires only multilinear tests now.

Since the pruning condition has been weakened, the analysis of
{\tt hull3d} needs to be changed.
Recall that the partition in line~3 is constructed by
recursive application of the 8-sectioning theorem.
At half the depth of recursion, we obtain an intermediate partition
of $Q$ with $O(\sqrt{r_j})$ subsets $Q_\ell'$ and corresponding cells $\gamma_\ell'$
where each subset has $O(n/\sqrt{r_j})$ points, and
every plane intersects at most $O(\sqrt{r_j}^{1-\eps})$ of
these cells $\gamma_\ell'$, for $\eps = 1-\log_8 7$.
Furthermore, for a fixed $\gamma_\ell'$,
every plane intersects at most
$O(\sqrt{r_j}^{1-\eps})$ of the cells $\gamma_i$ of
the final partition inside $\gamma_\ell'$. 

In the second paragraph of the proof of Theorem~\ref{thm:ch3d}, we do the analysis differently.  Consider a cell $\gamma_i$ of the partition,
which is contained in a cell $\gamma_\ell'$ of the intermediate
partition.  We claim that if (i)~$\gamma_\ell'$ is
strictly contained in $\D_k$ and (ii)~$\gamma_i$ is
strictly contained in $\gamma_\ell'$, then all points inside $\gamma_i$ are pruned.
To see this, notice that by (i), all points in $Q_\ell'$
are strictly below the upper hull of $Q$, and by (ii),
the defining points $\DEF(v)$ of any vertex $v$ of $\gamma_i$
are in $Q_\ell'$.  Thus, the points in $\DEF(v)$ are strictly
below the upper hull of $Q$, i.e., the upper hull of $Q-\DEF(v)$
is the same as the upper hull of $Q$.  As each vertex $v$ of $\gamma_i$
is strictly below the upper hull of $Q-\DEF(v)$, all points
inside $\gamma_i$ are indeed pruned.

At most $O(\sqrt{r_j}^{1-\eps})$ cells $\gamma_\ell'$ can intersect the boundary of $\D_k$.  For each of the $O(\sqrt{r_j})$ cells $\gamma_\ell'$
strictly contained in $\D_k$, 
at most $O(\sqrt{r_j}^{1-\eps}\log r_j)$ cells $\gamma_i$ inside
$\gamma_\ell'$ can intersect the $O(\log r_j)$ boundary facets of $\gamma_\ell'$.  Hence, the number of points in $S_k$ that remain in $Q$ after iteration $j$ is at most $\min\left\{|S_k|, O(\sqrt{r_j}^{1-\eps}\cdot n/\sqrt{r_j} + \sqrt{r_j}\cdot(\sqrt{r_j}^{1-\eps}\log r_j)\cdot n/r_j)\right\} = \min\left\{|S_k|, O((n/r_j^{\eps/2})\log r_j)\right\}$.  The rest of the proof is then the same, after readjusting $\eps$ by about a half.

\section{Other Applications}
\label{sec:other-applications}

We can apply our techniques to obtain instance-optimal algorithms for a number of geometric problems in the order-oblivious and random-order setting:

\begin{enumerate}
\item \emph{Off-line halfspace range reporting in 2-d and 3-d}: given a set $S$ of $n$ points and halfspaces, report the subset of points inside each halfspace.  Algorithms with $\Theta(n\log n+K)$ running time \cite{ChazelleGuibasLeeBIT85,ChanSICOMP00,AfshaniChanSODA09,ChaTsa} are known for total output size $K$.
\item \emph{Off-line dominance reporting in 2-d and 3-d}: given a set $S$ of $n$ red/blue points, report the subset of red points dominated by each blue point.  The problem has similar complexity as (1).
\item \emph{Orthogonal segment intersection in 2-d}: given a set $S$ of $n$ horizontal/vertical line segments, report all intersections between the horizontal and vertical segments, or count the number of such intersections.  The problem is known to have worst-case complexity $\Theta(n\log n + K)$ in the reporting version for output size $K$, and complexity $\Theta(n\log n)$ in the counting version~\cite{deBergvanKreveldBOOK,PreparataShamosBOOK}.
\item \emph{Bichromatic $L_\infty$-close pairs in 2-d}: given a set $S$ of $n$ red/blue points in 2-d, report all pairs $(p,q)$ where $p$ is red, $q$ is blue, and $p$ and $q$ have $L_\infty$-distance at most 1, or count the number of such pairs.  Standard techniques in computational geometry \cite{deBergvanKreveldBOOK,PreparataShamosBOOK} yield algorithms with the same complexity as in (3).
\item \emph{Off-line orthogonal range searching in 2-d}: given a set $S$ of $n$ points and axis-aligned rectangles, report the subset of points inside each rectangle, or count the number of such points inside each rectangle.  The worst-case complexity is the same as in (3).
\item \emph{Off-line point location in 2-d}: given a set $S$ of $n$ points and a planar connected polygonal subdivision of size $O(n)$, report the face in the subdivision containing each point.  Standard data structures \cite{deBergvanKreveldBOOK,PreparataShamosBOOK,SnoeyinkSURVEY} imply a worst-case running time of $\Theta(n\log n)$.
\end{enumerate}

For each of the above problems, it is not difficult to see that certain input sets are indeed ``easier'' than others, for example, if the horizontal segments and the vertical segments respectively lie inside two bounding boxes that are disjoint, then the orthogonal segment intersection problem can be solved in $O(n)$ time.

Note that although some of the above problems may be reducible to others in terms of worst-case complexity, the reductions may not make sense in the instance-optimality setting.  For example, an instance-optimal algorithm for a problem does not imply an instance-optimal algorithm for a restriction of the problem in a subdomain, because in the latter case, we are competing against algorithms that have to be correct only for input from this subdomain.


\subsection{A general framework for reporting problems}
\label{sec:app:rep}

\newcommand{\RR}{{\cal R}}

We describe our techniques for off-line reporting problems in a general framework.
Let $\RR\subset \R^d\times\R^{d'}$ be a relation for some constant dimensions $d$ and $d'$.  We say that a red point $p\in\R^d$ and a blue point $q\in\R^{d'}$ {\em interact\/} if $(p,q)\in\RR$.
We consider the {\em reporting\/} problem: given a set $S$ containing red points in $\R^d$ and blue points in $\R^{d'}$ of total size~$n$, report all $K$ interacting red/blue pairs of points in $S$.
(By scanning the output pairs, we can then collect the subset of all blue points that interact with each red point, in $O(K)$ additional time.)

We redefine $\entropy(S)$ as follows:
\begin{definition} \label{def:appl:rep} 
Given a cell $\gamma$ colored red (resp.\ blue),
we say that $\gamma$ is {\em safe for $S$\/} if every red (resp.\ blue) point in $\gamma$ interacts with exactly the same subset of blue (resp.\ red) points in $S$.  
We say that a partition $\Pi$ of $S$ is {\em{respectful}} if each subset $S_k$ in $\Pi$ is a singleton, or a subset of red points enclosed by a safe red simplex $\D_k$ for
$S$, or a subset of blue points enclosed by a safe blue simplex
$\D_k$ for $S$.
Define the \emph{structural entropy} $\entropy(S)$ of $S$ to be the minimum of $\entropy(\Pi) = \sum_k (|S_k|/n) \log(n/|S_k|)$ over all respectful partitions $\Pi$ of $S$.
\end{definition}

\begin{theorem}
$\OPT(S),\OPT\avg(S) \in\Omega(n(\entropy(S)+1)+K)$ for
the reporting problem in the multilinear decision tree model.
\end{theorem}
\begin{proof}
This follows from a straightforward modification
of the proofs of Theorems~\ref{thm:ch:lb} 
and~\ref{thm:randorder}. 
The main difference is that
we now keep two partition trees, one for the red (resp.\ blue)
points in $S$, with cells
colored red (resp.\ blue).  If a cell
$\Gamma(v)$ is safe for $S$, or if the number of red
(resp.\ blue) points in $\Gamma(v)$ drops below a constant, then
we make $v$ a leaf in the red (resp.\ blue) partition tree.
At the end, we argue that $v_p$ must be a leaf for every red (resp.\ blue) point
$p\in S$.  Otherwise, the red (resp.\ blue) cell $\Gamma(v_p)$ contains at least two red (resp.\ blue) points and is not safe, so we can move $p$ to another point inside $\Gamma(v_p)$ and change the answer.  The algorithm would be incorrect
on the modified input.
(The $\Omega(K)$ term in the lower bound is obvious.)
\end{proof}

For the upper-bound side, we assume the availability of three oracles 
concerning $\RR$, where $\alpha$ is some positive constant:
\begin{enumerate}
\item[(A)] A worst-case algorithm for the reporting problem that runs in $O(n\log n + K)$ time.
\item[(B)] A data structure with $O(n\log n)$ preprocessing time, such that we can report all $\kappa$ blue (resp.\ red) points in $S$ interacting with a query red (resp.\ blue) point in $O(n^{1-\alpha}+\kappa)$ time.
\item[(C)] A data structure with $O(n\log n)$ preprocessing time, such that we can test whether a query red or blue convex polyhedral cell $\gamma$ of size $a$
 is safe for $S$ in $O(an^{1-\alpha})$ time. 
\end{enumerate}

Note that we can reduce to preprocessing time in (B) 
to $O((n/m)\cdot m\log m)=O(n\log m)$ while increasing the query time to $O((n/m)\cdot m^{1-\alpha} + \kappa)$ for
any given $1\le m\le n$.
This follows from the grouping trick by \citeN{ChanDCG96b}:
namely, divide $S$ into $\lceil n/m\rceil$ subsets of
size $O(m)$ and build a data structure for each subset.
By setting $m=r^{1/\alpha}$,
we can then answer $r$ queries in total time
$O(n\log m + r \cdot (n/m)\cdot m^{1-\alpha} + \kappa)
\subseteq O(n\log r + r^{O(1)} + \kappa)$
for total output size $\kappa$.
Similarly, in (C), we can answer $r$ queries in total time 
$O(n\log r + ar^{O(1)})$.
The grouping trick is applicable because
the query problems in (B) and (C) are 
{\em decomposable\/}, i.e.,
the answer of a query for a union of subsets can be obtained
from the answers of the queries for the subsets.

We now solve the reporting problem by a variant 
of the \texttt{hull3d} algorithm in Section~\ref{sec:ch3d}:

\begin{algm}
\< \texttt{report}($Q$):\\
1.\' for $j=0,1\TO \lfloor\log(\de \log n)\rfloor$ do\\
2.\'\> partition the red points in $Q$ by Lemma~\ref{lem:part} to get $r_j=2^{2^j}$
	 subsets   
	$Q_1\TO Q_{r_j}$ \\\> and red cells $\gamma_1\TO\gamma_{r_j}$\\
3.\'\> for each $i=1$ to $r_j$ do\\
4.\'\>\> if $\gamma_i$ is safe for $Q$ then\\
5.\'\>\>\> let $Z_i$ be the subset of blue points in $Q$ that interact with\\\>\>\> an arbitrary red point in $Q_i$\\
6.\'\>\>\> output $Q_i\times Z_i$\\
7.\'\>\>\> prune all red points in $Q_i$ from $Q$\\
8.\'\> redo lines 2--7 with ``red'' and ``blue'' reversed\\
9.\' solve the reporting problem for the remaining set $Q$ directly
\end{algm}

The test in line~4 for each convex polyhedral cell $\gamma_i$
of size at most $O(\log r_j)$ can be done by 
querying the data structure in (C),
and line~6 can be done by querying the
data structure in (B); the cost of
$O(r_j)$ queries is $O(|Q|\log r_j + r_j^{O(1)})$ plus the
output size.
Line~9 can be done by the algorithm in (A).

\begin{theorem}
Given oracles (A), (B), and (C),
algorithm \mbox{\tt report}$(S)$ runs in $O(n(\entropy(S)+1)+K)$ time.
\end{theorem}
\begin{proof}
The analysis is as in the proof of Theorem~\ref{thm:ch3d}.
\end{proof}

The partition construction in line~2 can be done in the multilinear model, as described in Section~\ref{sec:multilinear}.  Whether
the rest of the algorithm works in the multilinear model
depends on the implementation of the oracles.

For orthogonal-type problems dealing with axis-aligned objects,
such as problems (2)--(5) in our list,
we can work instead in the comparison model.
We just replace simplices with axis-aligned boxes
in the definition of $\entropy(S)$, replace convex polyhedral cells
with axis-aligned boxes in oracle (C), and replace
partition trees with $k$-d trees in both the lower-bound
proof and the algorithm.

We can immediately apply our framework to the reporting versions 
of problems (1)--(5), after checking the oracle
requirements for (B) and (C) in each case:

\begin{enumerate}
\item \emph{Off-line halfspace range reporting in 2-d and 3-d}: 
For the design of the needed data structures, it suffices to consider just the lower halfspaces in the input.
Color the given points red, and map the given lower halfspaces to blue points by duality.  The data structure problem in (B) is just halfspace range reporting.  The data structure problem in (C) is equivalent to testing whether any of the $O(a)$ edges of
a query convex polyhedral cell intersects a given set of $n$ hyperplanes (lines in 2-d or planes in 3-d).  
This reduces to simplex range searching \cite{AgarwalEricksonSURVEY,MatousekDCG92} by duality;
known results achieve $O(n\log n)$ preprocessing time
and close to $O(an^{1-1/d})$ query time. 
It can be checked that the entire algorithm
is implementable in the multilinear model,
at least using a simpler randomized algorithm for (A)~\cite{ChanSICOMP00}. 

\item \emph{Off-line dominance reporting in 2-d and 3-d}: The data structure problem in (B) is just dominance reporting.  The data structure problem in (C) is equivalent to testing whether all the corners of a query box are dominated by the same number of points from a given $n$-point set.  This reduces to orthogonal range counting
\cite{AgarwalEricksonSURVEY,deBergvanKreveldBOOK,PreparataShamosBOOK};
although better data structures are known,
$k$-d trees are sufficient for our purposes,
with $O(n\log n)$ preprocessing time and $O(n^{1-1/d})$ query time.
The entire algorithm works in the comparison model.

\item \emph{Orthogonal segment intersection in 2-d}: Map each each horizontal line segment $\overline{(x,y)(x',y)}$ to a red point $(x,x',y)\in\R^3$ and each vertical line segment $\overline{(\xi,\eta)(\xi,\eta')}$ to a blue point $(\xi,\eta,\eta')\in\R^3$.  Each point in $\R^3$ is the image of a horizontal/vertical line segment.  The data structure problem in (B) for red queries corresponds to reporting all points from a given $n$-point set that lie
in a query range of the form
$\{(\xi,\eta,\eta')\in\R^3: ((x\le\xi\le x')\vee (x'\le\xi\le x)) 
\wedge ((\eta\le y\le \eta')\vee (\eta'\le y\le\eta))\}$ for some $x,x',y$.  
This reduces to 3-d orthogonal range reporting.
The data structure problem in (C) for red queries
corresponds to testing
whether a query box in $\R^3$ intersects any of the boundaries
of $n$ given ranges, where each range is of the form
$\{(x,x',y)\in\R^3: ((x\le\xi\le x')\vee (x'\le\xi\le x)) 
\wedge ((\eta\le y\le \eta')\vee (\eta'\le y\le\eta))\}$
for some $\xi,\eta,\eta'$.
This is an instance of 3-d orthogonal intersection searching
\cite{AgarwalEricksonSURVEY}, which reduces to orthogonal range searching in a higher dimension.
Again $k$-d trees are sufficient for our purposes. 
Blue queries are symmetric.
The entire algorithm works in the comparison model.

\item \emph{Bichromatic $L_\infty$-close pairs in 2-d}: The problem in (B) corresponds to reporting all points of a given point set that lie inside a query square of side length~2.  
This is an instance of orthogonal range reporting.
The problem in (C) corresponds to testing whether a query box
intersects any of the edges of $n$ given squares of side length~2.
This is an instance of orthogonal intersection searching.
Note that here the resulting algorithm requires a slight
extension of the comparison model, to include tests
of the form $x_i\le x_j'+a$ mentioned in Remark~\ref{rmk:maxima:lb}, which are allowed in the lower-bound proof.

\item \emph{Off-line orthogonal range reporting in 2-d}: Color the given points red, and map each rectangle 
$[\xi,\xi']\times [\eta,\eta']$ to a blue point $(\xi,\xi',\eta,\eta')\in\R^4$.  Every point
in $\R^4$ is the image of a rectangle.  
The problem in (B) for red queries corresponds to 2-d rectangle stabbing,
i.e., reporting all rectangles, from a given set of $n$ rectangles, that contain a query point.
The problem in (B) for blue queries corresponds to 2-d orthogonal range reporting.
The problem in (C) for red queries corresponds to deciding
whether a query box in $\R^2$ intersects any of the edges
of $n$ given rectangles.
The problem in (C) for blue queries corresponds to deciding
whether a query box in $\R^4$ intersects any of the boundaries
of $n$ given ranges, where each range is of the form
$\{(\xi,\xi',\eta,\eta')\in\R^4: 
((\xi\le x\le\xi')\vee (\xi'\le x\le\xi)) \wedge
((\eta\le y\le\eta')\vee (\eta'\le y\le\eta))\}$ for some $x,y$.
  All these data structure problems reduce to orthogonal range or intersection searching.  Again the algorithm works
in the comparison model.
\end{enumerate}

\subsection{Counting problems}

Our framework can also be applied to \emph{counting problems}, where we simply want the total number of interacting red/blue pairs.  
We just change oracle (A) to a counting algorithm without the $O(K)$
term, and oracle (B) to counting data structures without the $O(\kappa)$ term.  In line~5 of the algorithm we compute
$|Z|$, and in line~6 we add $|Q_i|\times |Z|$ to a global counter.
The same lower- and upper-bound proofs yield an $\Theta(n(\entropy(S)+1))$ bound.
The new oracle requirements are satisfied for 
(3)~orthogonal segment intersection counting, 
(4)~bichromatic $L_\infty$-close pairs, and 
(5)~off-line orthogonal range counting.

We can also modify the algorithm to return \emph{individual
counts}, i.e., compute the number of red points that interact with each blue point and the number of blue points that interact
with each red point.  Here, we need to not only strengthen oracle (A) to produce individual counts, but also modify 
oracle (B) to the following:
\begin{enumerate}
\item[(B)]  A data structure with $O(n\log n)$ preprocessing time,
forming a collection of canonical subsets of total size $O(n\log n)$, such that we can express the subset of all blue (resp.\ red) points in $S$ interacting with a query red  (resp.\ blue) point, as a union of
$O(n^{1-\alpha})$ canonical subsets, in 
$O(n^{1-\alpha})$ time.
\end{enumerate}
As before, the grouping trick can be used to reduce the preprocessing time and total size of the canonical subsets.
In line~4 of the algorithm
we express $Z$ as a union of canonical subsets.
In line~5 we add $|Z|$ to the counter of each
red point in $Q_i$ and add $|Q_i|$ to the counter of
each canonical subset for $Z$.
At the end of the loop in
lines 3--7, we make a pass over each canonical subset
and add its counter value to the counters of its blue points,
before resetting the counter of the canonical subset.
Line~8 is similar.
The analysis of the running time remains the same.
The strengthened oracle requirements are satisfied
for problems (3), (4), and (5) by known orthogonal
range searching results.

\subsection{Detection problems?}

We can also consider \emph{detection problems} where we simply want to decide whether there exists an interacting red/blue pair.  Here, we redefine $\entropy(S)$ by redefining ``safe'': a red (resp.\ blue) cell $\gamma$ is now considered {\em safe for $S$\/} if no red (resp.\ blue) point in $\gamma$ interacts with any blue (resp.\ red) points in $S$.  We change oracles (A) and (B) to 
analogous detection algorithms and data structures, without the $O(K)$ and $O(\kappa)$ terms.

The proof of the upper bound $O(n(\entropy(S)+1))$ is the same, but unfortunately the proof of the lower bound $\Omega(n(\entropy(S)+1))$ only works
for instances with a {\sc no} answer: at the end, if $v_p$ is not a leaf for some red (resp.\ blue) point $p\in S$, then $\Gamma(v_p)$
contains at least two red (resp.\ blue) points and is not safe, so we can move $p$ to some point inside $\Gamma(v_p)$ and change the answer from {\sc no} to {\sc yes}.

{\sc yes} instances are problematic, but this is not a weakness of our technique but of the model: on every input set $S$ with a {\sc yes} answer, $\OPT(S)$ is in fact $O(n)$.  To see this, consider an input set $S$ for which there exists an interacting pair $(p,q)$.  An algorithm that is ``hardwired'' with the ranks of $p$ and $q$ in $S$ with respect to, say, the $x$-sorted order of $S$ can first find $p$ and $q$ from their ranks by linear-time selection, verify that $p$ and $q$ interact in constant time, and return {\sc yes} if true or run a brute-force algorithm otherwise.  Then on every permutation of this particular set $S$, the algorithm always takes linear time.  Many problems admit $\Omega(n\log n)$ worst-case lower bounds even when restricted to {\sc yes} instances, and for such problems, instance optimality in the order-oblivious setting is therefore not possible on all instances.

\subsection{Another general framework for off-line querying problems}

\newcommand{\MM}{{\cal M}} 

We now study problems from another general framework.  Let $\MM$ be a mapping from points in $\R^d$ to ``answers'' in some space
for some constant $d$ (the answer $\MM(q)$ of a point $q\in\R^d$ may or may not have constant size depending on the context).  We consider the following {\em off-line querying\/} problem: given a set $S$ of $n$ points in $\R^d$, compute $\MM(q)$ for every $q\in S$.  Let $K$ denote the total size of the answers.

We redefine $\entropy(S)$ by redefining ``safe'': 

\begin{definition} \label{def:offline:query} 
We say that a cell $\gamma$ is {\em safe\/} if every point $q$ in $\gamma$ has the same answer $\MM(q)$.
We say that a partition $\Pi$ of $S$ is {\em{respectful}} if each subset $S_k$ in $\Pi$ is a singleton, or a subset of points enclosed by a safe simplex $\D_k$.
Define the \emph{structural entropy} $\entropy(S)$ of $S$ to be the minimum of $\entropy(\Pi) = \sum_k (|S_k|/n) \log(n/|S_k|)$ over all respectful partitions $\Pi$ of $S$.
\end{definition}

\begin{theorem}
$\OPT(S),\OPT\avg(S) \in\Omega(n(\entropy(S)+1)+K)$ for
the off-line querying problem in the multilinear decision tree model.
\end{theorem}
\begin{proof}
This follows from a straightforward modification
of the proofs of Theorems~\ref{thm:ch:lb} 
and~\ref{thm:randorder}. 
As before, if a cell $\Gamma(v)$ is safe, then we make $v$ a leaf.
At the end, we argue that $v_p$ must be a leaf for every
$p\in S$.  Otherwise, $\Gamma(v_p)$ is not safe, so we can move $p$ to another point inside $\Gamma(v_p)$ and change the answer.  The algorithm would be incorrect on the modified input.
\end{proof}

The above lower bound holds even if we ignore the cost
of preprocessing $\MM$.  Furthermore, the test functions
are only required to be multilinear with respect to $S$,
not~$\MM$. 

For the upper-bound side, we assume that $\MM$ has been preprocessed in an oracle data structure supporting the
following types of queries:
\begin{enumerate}
\item[(A)] Given $q\in\R^d$, we can compute $\MM(q)$ in $O(\log m+\kappa)$ worst-case time for output size~$\kappa$,
where $m$ is a
parameter describing the size of $\MM$.
\item[(C)] Given a convex polyhedral cell $\gamma$ of size $a$, 
we can test whether $\gamma$ is safe in $O(am^{1-\alpha})$ time.
\end{enumerate}

The algorithm is simpler this time.
Instead of using a $2^{2^j}$ progression, we can use a more straightforward $b$-way recursion, for a sufficiently large
constant $b$ (the resulting recursion tree mimics the tree $\TT$ from the lower-bound proof in Theorem~\ref{thm:ch:lb}, on purpose):

\begin{algm}
\< {\tt off-line-queries}$(Q,\Gamma)$, where $Q\subset\Gamma$:\\
1.\' if $|Q|$ drops below $n/m^\de$ then compute
the answers directly and return\\
2.\' partition $Q$ by Lemma~\ref{lem:part} to get $b$
       subsets $Q_1\TO Q_b$  and cells $\gamma_1\TO\gamma_b$\\
3.\' for $i=1$ to $b$ do\\
4.\'\> if $\gamma_i\cap\Gamma$ is safe then\\
5.\'\>\> compute $\MM(q)$ for an arbitrary point  $q\in\gamma_i\cap\Gamma$\\
6.\'\>\> output $\MM(q)$ as the answer for all points in $Q_i$\\
7.\'\> else {\tt off-line-queries}$(Q_i,\gamma_i\cap\Gamma)$
\end{algm}

We call {\tt off-line-queries}$(S,\R^d)$ to start.  Line~1 takes $O(|Q|\log m+\kappa)$ time for output size $\kappa$ by querying
the data structure for (A); note that each point in $Q$ in this case has participated in $\Omega(\log m)$ levels of the recursion, and we can account for the first term by charging each point unit cost for every level it participates in.  Line~2 takes $O(|Q|)$ time for a constant $b$ by known constructions of
Matou\v sek's partition theorem~\cite{MatousekDCG92}
(or alternatively recursive application of the 4- or 8-sectioning theorem in the 2-d or 3-d case).  The test in line~4
takes $O(m^{1-\alpha}\polylog m)$ time by querying the data structure for (C), since the convex polyhedral
cell $\gamma_i\cap\Gamma$ has at most $O(\log m)$ facets
(and thus $O(\!\polylog m)$ size).  
As the tree has $O(m^\de)$ nodes, the cost of line~4 is negligible, since its total
over the entire recursion tree is sublinear in $m$
by choosing a sufficiently small constant $\de<\alpha$.  Line~5 takes $O(\log m +\kappa)$ time for output size $\kappa$, by (A); the $O(\log m)$ term
is again negligible, since its total over
the entire recursion tree is sublinear in $m$.


\begin{theorem}\label{thm:offline}
After $\MM$ has been preprocessed for (A) and (C),
algorithm {\tt off-line-queries}$(S,\R^d)$ runs in $O(n(\entropy(S)+1)+K)+o(m)$ time for total output size~$K$.  
\end{theorem}
\begin{proof}
Let $n_j$ be number of points in $S$ that {\em survive\/} level $j$, i.e., participate in subsets $Q$ at level $j$ of the recursion.  The total running time for the off-line problem is asymptotically bounded by $\sum_j n_j$, ignoring a $o(m)$
extra term.  

Let $\Pi$ be any respectful partition of $S$.
Consider a subset $S_k$ in $\Pi$.
Let $\D_k$ be a safe simplex enclosing $S_k$.  Fix a level $j$.  Let $Q_i$'s and $\gamma_i$'s be the subsets $Q$ and cells $\gamma$ at level $j$.
Each $Q_i$ has size at most $n/\Theta(b)^j$.
The number of $\gamma_i$'s that intersect the boundary of $\D_k$ is at most $O(b^{1-\eps})^j$.
Thus, the number of points in $S_k$ that survive level $j$ is at most $\min\left\{|S_k|, O(b^{1-\eps})^j \cdot n/\Theta(b)^j \right\}$.
Since the $S_k$'s cover the entire point set, with a double summation we have, for a sufficiently large constant $b$,
\begin{eqnarray*}
\sum_j n_j &\le&
\sum_j \sum_k \min\left\{|S_k|, n/\Theta(b)^{\eps j}\right\} \\
&=& \sum_k \sum_j \min\left\{|S_k|,n/\Theta(b)^{\eps j}\right\} \\
&\subseteq & \sum_k O\left(|S_k| (\log(n/|S_k|) + 1)\right)\\
  &=& O(n(\entropy(\Pi)+1)),
\end{eqnarray*}
which yields the theorem.
\end{proof}

For orthogonal-type problems, we can work instead in the comparison model.  We just replace simplices with axis-aligned boxes
in the definition of $\entropy(S)$, replace convex polyhedral cells
with axis-aligned boxes in oracle (C), and replace
partition trees with $k$-d trees in both the lower-bound
proof and the algorithm.

We can apply our framework to solve problem (6):
\begin{enumerate}
\item[(6)]
\emph{Off-line point location in 2-d}: 
For (A), data structures for
planar point location with $O(\log m)$ worst-case query time are known, with $O(m)$ preprocessing time and space \cite{KirkpatrickSICOMP83,ChazelleDCG91,SnoeyinkSURVEY}.  The data structure problem in (C) is equivalent to testing whether any of the
$O(a)$ edges of a query polygon intersect
the given polygonal subdivision.  This reduces to ray shooting (or segment emptiness) queries in the subdivision, for which known results~\cite{ChazelleEdelsbrunnerGrigniALGMCA94}
achieve $O(\log m)$ query time,
with $O(m)$ preprocessing time and space.  
For this problem, each answer has constant size, so the $O(K)$
and $O(\kappa)$ terms can be omitted.
The total running time is $O(n(\entropy(S)+1))$, even if
preprocessing time is included, for a subdivision
of size $m=O(n)$. 
It can be checked that the entire algorithm works in 
the multilinear model.
\end{enumerate}

\subsection{On-line querying problems and
distribution-sensitive data structures}

In the general framework of the preceding section, we can also consider the following {\em on-line querying\/} problem: given a set $S$ of $n$ points in $\R^d$, build a data structure so that we can compute $\MM(q)$ for any query point $q\in\R^d$, while trying to minimize the {\em average query cost over all $q\in S$}.

Our off-line lower bound states that the total time
required to answer queries for all $n$ points in $S$
is $\Omega(n(\entropy(S)+1))$.  This immediately implies
that the average query time over all $q\in S$ must
be $\Omega(\entropy(S)+1)$.  (In contrast, lower bounds for the on-line problem do not necessarily translate to lower bounds for the off-line problem.)

On the other hand, our algorithm for off-line queries can
be easily modified to give a data structure for on-line queries.
We just build a data structure corresponding to the recursion tree generated by {\tt off-line-queries}$(S,\R^d)$,
in addition to the data structure for (A) and (C).
It can be shown that with such a data structure,
the average query time over
all $q\in S$ is $O(\entropy(S)+1)$ (more details are given below).

We can extend the on-line querying
problem to the setting where each point in $S$ is
weighted and the goal is to bound the weighted
average query time over the query points in $S$.
Even more generally, we can consider the setting where
$S$ is replaced by a (possibly continuous)
probability distribution and the goal is to
bound the expected query time for a query point
randomly chosen from $S$.
We now provide more details
for the changes needed in this most general setting.

We first redefine $\entropy(S)$ for a probability
distribution $S$.  

\begin{definition}
We say that a cell $\gamma$ is {\em safe\/} if every point $q$ in $\gamma$ has the same answer $\MM(q)$.
A partition $\Pi$ of $\R^d$ into regions is
{\em respectful\/} if each region $S_k$ of $\Pi$ can be enclosed by
a safe simplex $\D_k$. 
Define the {\em structural entropy} $\entropy(S)$
of a probability distribution $S$ to be the minimum of
$\entropy(\Pi) = \sum_k \mu_S(S_k)\log (1/\mu_S(S_k))$ 
over all respectful partition $\Pi$ of $\R^d$,
where $\mu_S$ denotes the probability measure corresponding to $S$.
\end{definition}

We need a continuous version of Lemma~\ref{lem:part},
which follows by straightforward modification to the proof of the
partition theorem \cite{MatousekDCG92}
(or alternatively, recursive application of 
the 4- or 8-sectioning theorem in the 2-d or 3-d case).

\begin{lemma}\label{lem:part:cont}
For any probability measure in $\R^d$ and $1\le r\le n$
for any constant~$d$, we can partition $\R^d$ into
$r$ (not necessarily convex or connected) polyhedral
regions $Q_1\TO Q_r$ each with measure $\Theta(1/r)$ and
with $r^{O(1)}$ (or fewer) facets,
and find $r$ convex polyhedral cells $\gamma_1\TO\gamma_r$
each with $O(\log r)$ (or fewer) facets, such that $Q_i$
is contained in $\gamma_i$, and
every hyperplane intersects at most $O(r^{1-\eps})$ cells.  Here, $\eps>0$ is a constant that depends only on $d$.
\end{lemma}

The lower bound proof is easier for on-line problems,
so we present the simplified proof in full below.  Because the input to a query algorithm is just 
a single point (unlike in the off-line setting where
we could perform a test that has more than one query points
as arguments),
multilinear tests can now be replaced with linear tests.

\begin{theorem}
Any algorithm for the on-line querying problem
requires $\Omega(\entropy(S)+1+\kappa)$ expected query time
for output size $\kappa$,
for any probability distribution $S$ in the linear decision tree
model.
\end{theorem}
\begin{proof}
We define a {\em partition tree\/} $\TT$ as follows: Each node $v$ stores
a pair $(Q(v),\Gamma(v))$, where $Q(v)$ is a region enclosed inside
a convex polyhedral cell $\Gamma(v)$.
%
%
The root stores $(\R^d,\R^d)$.
If $\mu_S(Q(v))$ drops below $1/m^\de$ or $\Gamma(v)$ is safe, then $v$ is a leaf.
Otherwise, apply Lemma~\ref{lem:part} with $r=b$ to 
the restriction of $\mu_S$ to $Q(v)$, and 
obtain regions $Q_1\TO Q_b$ and cells $\gamma_1\TO \gamma_b$.
For the children $v_1\TO v_b$ of $v$, set $Q(v_i)=Q_i\cap Q(v)$
and $\Gamma(v_i)=\gamma_i\cap\Gamma(v)$.
For a node $v$ at depth $j$ of the tree $\TT$ we then have $\mu_S(Q(v)) \ge 1/\Theta(b)^j$, and consequently,
the depth $j$ is in $\Omega(\log_b (1/\mu_S(Q(v))))$.

Let $\PiPart$ be the partition formed by the cells $\Gamma(v)$ at the leaves $v$ in $\TT$.
Let $\PiPartt$ be a refinement of this partition after
triangulating the leaf cells.
Note that the subpartitioning of a leaf cell at depth $j$ causes the entropy to increase
by at most $O(\mu_S(Q(v))\log (j\log b))\subseteq O(\mu_S(Q(v))\log\log (1/\mu_S(Q(v))))$ for any constant $b$.
So $\entropy(\PiPartt)\in\Theta(\entropy(\PiPart))$.  Clearly, $\PiPartt$ is respectful.

The adversary constructs a bad query point $q$ as follows.  During the simulation, we maintain a node $v_q$ in $\TT$, where
the only information the algorithm knows about $q$ is that $q$ lies inside $\Gamma(v_q)$.

Suppose that the simulation encounters a test to determine
which side $q$ lies inside a hyperplane $h$.  We do the following:
\begin{itemize}
\item We reset $v_q$ to one of its children at random, where
each child $v_q'$ is chosen with probability
$\mu_S(Q_q')/\mu_S(Q_q)$ (which is in $\Theta(1/b)$).  
If $\Gamma_q'$ does not intersect $h$, then
the comparison is resolved.  Otherwise, we repeat.
\end{itemize}
The probability that the comparison is not resolved
in a single step
is at most $O(b^{1-\eps})\cdot \Theta(1/b)$, which 
can be made less than $1/2$ for a sufficiently large constant $b$.

Let $T$ be the number of tests made.
Let $D$ be the depth of $v_q$ at the end.
For $i\le T$, let $T_i=1$ and
$D_i$ be the number of steps needed to resolve the $i$-th test;
then $\Ex[D_i]\le 2$.
For $i>T$, let $T_i=D_i=0$.
Since $\Ex[2T_i-D_i]\ge 0$ for all~$i$, by linearity of expectation
$\Ex[D]=\Ex[\sum_i D_i]\le 2\, \Ex[\sum_i T_i] = 2\,\Ex[T]$.

At the end of the algorithm, $v_q$ must be a leaf.
Thus,
\begin{eqnarray*}
 \Ex[T] &\ge& \Omega(\Ex[D])\\
&\subseteq& \Omega\left(\sum_{\mbox{\rm\scriptsize leaf }v} 
\mu_S(Q(v))\log(1/\mu_S(Q(v)))\right)\\
&\subseteq& \Omega(\entropy(\PiPart))\\
&\subseteq& \Omega(\entropy(\PiPartt))\\
&\subseteq& \Omega(\entropy(S)).
\end{eqnarray*}

At the end, we can assign $q$ to a random point in $Q(v_q)$
chosen from the distribution $S$.
Then $q$ satisfies precisely the probability distribution $S$
(in other words, the adversary have not acted adversarily after all).
Thus, $\Ex[T]$ is the expected query time for a query point
randomly chosen from the distribution $S$.
\end{proof}

For the upper-bound side, below is the pseudocode for
the preprocessing algorithm and the
query algorithm, which are derived from our previous
algorithm {\tt off-line-queries}; here, $b$ is a sufficiently
large constant. 
The query algorithm assumes
an oracle data structure for (A) (oracle (C) is only needed
in the preprocessing algorithm). 

\begin{algm}
\< {\tt preprocess}$(Q,\Gamma)$:\\
1.\' if $\mu_S(Q) < 1/m^\de$ then return\\
2.\' apply Lemma~\ref{lem:part:cont} to the restriction of $\mu_S$
to $Q$ \\ to get $b$ regions $Q_1\TO Q_b$ and cells $\gamma_1\TO\gamma_b$\\
3.\' for $i=1$ to $b$ do\\
4.\'\> if $\gamma_i\cap\Gamma$ is safe then
store $\MM(q)$ for an arbitrary point $q\in\gamma_i\cap\Gamma$\\
5.\'\> else {\tt preprocess}$(Q_i\cap Q,\gamma_i\cap\Gamma)$
\end{algm}

\begin{algm}
\< {\tt on-line-query}$(q,Q,\Gamma)$:\\
1.\' if $\mu_S(Q) < 1/m^\de$ then compute $\MM(q)$ directly and return\\
2.\' locate the region $Q_i$ containing $q$\\
3.\' if $\gamma_i\cap\Gamma$ was marked as safe then
look up the stored answer and return\\
4.\' else {\tt on-line-query}$(q,Q_i\cap Q,\gamma_i\cap\Gamma)$
\end{algm}

The space of the tree generated by
{\tt preprocess}$(\R^d,\R^d)$ is only $O(m^\de)$, and so
the space for the data structure for (A) dominates.
We will not focus on the preprocessing time,
which depends on the construction
time for Lemma~\ref{lem:part:cont}, which in turn
depends on the distribution $S$.
The preprocessing can be done efficiently, for example,
for a discrete $n$-point distribution and for
many other distributions.

In algorithm {\tt on-line-query},
line~1 takes $O(\log m)$ time by (A); note that this
case occurs only when the query point $q$
participates in $\Omega(\log m)$ levels of the recursion,
and we can account for the cost by charging one unit
to each level of the recursion.  
Line~2 takes $O(1)$ time for a constant $b$.
We can adapt the previous proof of Theorem~\ref{thm:offline}
for the rest of the query time analysis:

\begin{theorem}
After $\MM$ has been preprocessed for (A) and {\tt preprocess}$(\R^d,\R^d)$ has been executed,
algorithm {\tt on-line-query}$(q,\R^d,\R^d)$ runs in $O(\entropy(S)+1+\kappa)$ expected time for output size~$\kappa$,
for a query point $q$ randomly chosen from the distribution~$S$.  
\end{theorem}
\begin{proof}
Let $n_j$ be 1 if the query point participates at level $j$ of the recursion, and 0 otherwise.
Then the query time is asymptotically
bounded by $\sum_j n_j$.

Let $\Pi$ be any respectful partition of $S$.
Consider a region $S_k$ in $\Pi$, enclosed in a safe simplex~$\D_k$.
Fix a level $j$.  Let $Q_i$'s and $\Gamma_i$'s be the regions and cells 
at level $j$ of the recursion tree.
Each $Q_i$ satisfies $\mu_S(Q_i)\le 1/\Theta(b)^j$.
The number of $\Gamma_i$'s that intersect the boundary of $\D_k$ is at most $O(b^{1-\eps})^j$.
Thus, $\Pr[(n_j=1)\wedge (q\in S_k)] \le \min\{\mu_S(S_k),\, O(b^{1-\eps})^j \cdot 1/\Theta(b)^j\}$ for a point $q$
randomly chosen from the distribution~$S$.
Since the $S_k$'s cover $\R^d$, with a double summation we have, for a sufficiently large constant $b$,
\begin{eqnarray*}
\Ex\left[\sum_j n_j\right] &\le&
\sum_j \sum_k \min\{\mu_S(S_k), 1/\Theta(b)^{\eps j}\} \\
&=& \sum_k \sum_j \min\{\mu_S(S_k), 1/\Theta(b)^{\eps j}\} \\
&\subseteq & \sum_k O\left(\mu_S(S_k) (\log(1/\mu_S(S_k)) + 1)\right)\\
  &=& O(\entropy(\Pi)+1),
\end{eqnarray*}
which yields the theorem.
\end{proof}

For orthogonal-type problems, we can again work in the comparison
model, by replacing simplicial and convex polyhedral cells
with axis-aligned boxes.

We can apply our framework to on-line versions of
several problems:

\begin{enumerate}
\item[(6)] \emph{On-line point location queries in 2-d}: We immediately obtain optimal $O(\entropy(S)+1)$ expected query cost, with an $O(m)$-space data structure for a subdivision of size $m$,
for any given distribution $S$.  The query algorithm works
in the linear decision tree model.
This on-line point location result is known before \cite{AryaMalamatosTALG07,AryaMalamatosSICOMP07,ColleteDujmovicSODA08,IaconoCGTA04} (some of these previous work even optimize the constant factor in the query cost).

\item[(1)] \emph{On-line halfspace range reporting queries in 2-d and 3-d}: Here, we map query lower halfspaces to points by duality.  For (A), data structures for 2-d and 3-d halfspace range reporting with
$O(\log m +\kappa)$ worst-case time are known, with $O(m)$ space \cite{ChazelleGuibasLeeBIT85,AfshaniChanSODA09}.
We thus obtain optimal $O(\entropy(S)+1+\kappa)$ expected query cost for output size $\kappa$, with an $O(m)$-space data structure for a given 2-d or 3-d $m$-point set, for any given distribution~$S$.  
The query algorithm works in the linear decision tree model.
This result is new.

\item[(2)] \emph{On-line dominance reporting queries in 2-d and 3-d}: The story is similar to halfspace range reporting.
The query algorithm now works in the comparison model.

\item[(4)] \emph{On-line orthogonal range reporting/counting queries in 2-d}: Here, we map query rectangles to points in 4-d as in Section~\ref{sec:app:rep}.  For (A),
data structures for 2-d orthogonal range reporting
with $O(\log m+\kappa)$ worst-case query time are known, with
$O(m\log^\eps m)$ space; and 
data structures for 2-d orthogonal range counting
with $O(\log m)$ worst-case query time are known,
with $O(m)$ space~\cite{ChazelleSICOMP88}.  
For reporting,
we thus obtain optimal $O(\entropy(S)+1+\kappa)$ expected query cost for output size $\kappa$, with an $O(m\log^\eps m)$-space data structure for a given 2-d $m$-point set, for any given
distribution~$S$;
for counting, we get optimal $O(\entropy(S)+1)$ expected
query cost with an $O(m)$-space data structure.
The query algorithm works in the comparison model.  This result is apparently new, as it extends Dujmovi\'c, Howat, and Morin's result on 2-d dominance counting~\cite{DujmovicHowatSODA09} and unintentionally answers one of their main open problems (and at the same time improves their space bound from $O(m\log m)$ to $O(m)$).
\end{enumerate}

\begin{remark}
Some months after the appearance of the
conference version of the present paper,
a similar general technique for distribution-sensitive data structures was rediscovered by \citeN{BoseOddsOn}.
\end{remark}

\section{Discussion}
\label{sec:discussion}

Although we have argued for the order-oblivious form of instance optimality, we are not denigrating adaptive algorithms that exploit the order of the input.  Indeed, for some geometric applications, the input order may exhibit some sort of locality of reference which can speed up algorithms.  There are various parameters that one can define to address this issue, but it is unclear how a unified theory of instance optimality can be developed for order-dependent algorithms.

We do not claim that the algorithms described here are the best in practice, because of possibly larger constant factors (especially those that use Matou\v sek's partition trees), although some variations of the ideas might actually be useful.  In some sense, our results can be interpreted as a theoretical explanation for why heuristics based on bounding boxes and BSP trees perform so well (for example, see~\cite{AndrewsSnoeyinkSDH94} on experimental results for the red/blue segment intersection problem), as many of our
instance-optimal algorithms prune input based on bounding boxes
and spatial tree structures.

Note that specializations of our techniques to 1-d can also lead to order-oblivious instance-optimal results for the multiset-sorting problem and the problem of computing the intersection of two (unsorted) sets.  Adaptive algorithms for similar 1-d problems (e.g.,~\cite{MunroSpira}) were studied in different settings from ours.

Not all standard geometric problems admit nontrivial instance-optimal results in the order-oblivious setting.  For example, computing the Voronoi diagram of $n$ points or the trapezoidal decomposition of $n$ disjoint line segments, both having $\Theta(n)$ sizes, requires $\Omega(n\log n)$ time for every point set by the naive information-theoretic argument.  Computing the ($L_\infty$-)closest pair for a {\em monochromatic\/} point set requires $\Omega(n\log n)$ time for every point set by our adversary lower-bound argument.

An open problem is to strengthen our lower bound proofs to allow for a more general class of test functions beyond multilinear functions, for example, arbitrary fixed-degree algebraic functions.

It remains to see how widely applicable the concept of instance optimality is.  To inspire further work, we mention the following geometric problems for which we currently are unable to obtain instance-optimal results: 
\begin{enumerate}
\item[(a)]~reporting all intersections between a set of disjoint red (nonorthogonal) line segments and a set of disjoint blue line segments in 2-d; 
\item[(b)]~computing the $L_2$- or $L_\infty$-closest pair between a set of red points and a set of blue points in 2-d; 
\item[(c)]~computing the diameter or the width of a 2-d point set; 
\item[(d)]~computing the lower envelope of a set of (perhaps disjoint) line segments in 2-d.
\end{enumerate}

Finally, we should mention that all our current results concern at most logarithmic-factor improvements.
Obtaining some form of instance-optimal results for problems with $\omega(n\log n)$ worst-case complexity (for example, off-line triangular range searching, 3SUM-hard problems, \ldots) would be even more fascinating.

\appendix

\section*{APPENDIX}
\section{An Alternative Proof for 2-d Maxima}
\label{sec:maxima:alt}

\newcommand{\FF}{{\cal F}}

In this appendix, we describe an alternative approach
to the 2-d maxima problem, which uses
a new definition of a difficulty measure and 
a vastly different lower-bound proof based on an interesting
encoding argument.  This approach is more specialized
and does not seem to work for 3-d maxima or other problems, but the lower-bound proof has the advantage of being generalizable
to nondeterministic algorithms.

We begin by defining a measure of difficulty $\FF(S)$ specific
to the 2-d maxima problem, which is seemingly different from the structural entropy $\entropy(S)$ defined previously.
The new definition appears simpler in the sense that we do not need to take the minimum over all partitions but measure the contribution of each point directly, but as a byproduct of our analyses, $\FF(S)$ is asymptotically equivalent to $n\entropy(S)$ (which is why we do not give it a name).

\begin{definition}\label{def:F}
Given a point set $S$, let $q_1\TO q_h$ denote the maximal points of $S$ from left to right, with $q_0=(-\infty,\infty)$
and $q_{h+1}=(\infty,-\infty)$. 
Given a point $p\in S$, let $q_i\TO q_\ell$ be all the maximal points that dominate $p$.  
Define $F(p)$ to be the subset of all points in $S$ in the slab $(q_{i-1}.x,q_{\ell+1}.x)\times\R$, where we use $p.x$ and $p.y$ to denote the $x$- and $y$-coordinates of $p$.
Define $\FF(S) = \sum_{p\in S} \log (n/|F(p)|)$.
(See Figure~\ref{fig:DefinitionOfFFS}.)
\end{definition}

For the upper-bound side, we use the same algorithm 
and an analysis similar to before:

\begin{figure}
\centering
\begin{minipage}[c]{.45\linewidth}
\begin{tikzpicture}[scale=.85]
\DefineInstanceMaximaTwoD%
\draw [output] (0.1,0.1) |- (m1)  |- (m2) |- (m3) |- (0.1,0.1) ; 
\draw [maxima] (m0) circle (2pt) (m4) circle (2pt);
\draw [above right] (m0) node {$q_0$} (m4) node  {$q_4$};
\DrawInstanceMaximaTwoD%
\path 
node[below] at (p1) {$p_1$}
node[below] at (p2) {$p_2$}
node[below] at (p3) {$p_3$}
node[below] at (p4) {$p_4$}
node[below] at (p5) {$p_5$}
node[below] at (p6) {$p_6$}
node[above] at (p7) {$p_7$}
node[above] at (p8) {$p_8$}
node[above] at (p9) {$p_9$}
;
\draw[color=blue,dashed] (p2)
 edge [<-] (m1)  
 edge [<-] (m2)   
 edge [<-] (m3);
\end{tikzpicture}
\end{minipage}\begin{minipage}[c]{.5\linewidth}
\caption{Definition of $\FF(S)$: In this instance,
points $p_1,p_2,p_3$ are dominated by $q_1,q_2,q_3$, and so $|F(p_1)|=|F(p_2)|=|F(p_3)|=n=12$. 
Points $p_4,p_5,p_6,q_2$ are dominated only by $q_2$,
and so $|F(p_4)|=|F(p_5)|=|F(p_6)|=|F(q_2)|=7$.
Similarly $|F(p_7)|=|F(p_8)|=|F(p_9)|=|F(q_3)|=4$
and $|F(q_1)|=7$.  Thus, $\FF(S)=3\log\frac{12}{12} + 5\log\frac{12}{7}+4\log\frac{12}{4}$.
\label{fig:DefinitionOfFFS}
}
\end{minipage}
\end{figure}
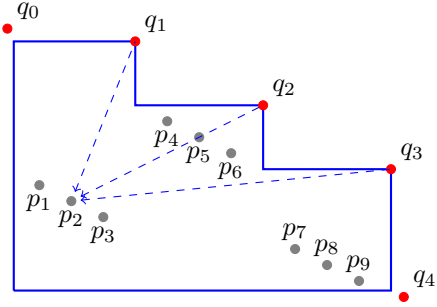

\begin{theorem}
Algorithm {\tt maxima2d}$(S)$ from Section~\ref{sec:maxima:ub} runs in $O(\FF(S)+n)$ time.
\end{theorem}
\begin{proof}
We proceed as in the proof of Theorem~\ref{thm:maxima:ub}, but a simpler argument replaces the second paragraph:
Fix a point $p\in S$.  Let $q_i\TO q_\ell$ be all the maximal points that dominate $p$.  Fix a level $j$.
If $|F(p)| > \down{n/2^j}$, then by (i), that some maximal point from $\{q_i\TO q_\ell\}$ must been discovered, and by (ii),
this implies that $p$ does not survive level $j$.
Thus, $p$ can survive only for $O(\log (n/|F(p)|)+1)$ levels.  We can asymptotically bound the running time by $\sum_j n_j
\in O(\sum_p (\log (n/|F(p)|)+1)) = O(\FF(S)+n)$.
\end{proof}

For the lower-bound side, we first consider
a slightly strengthened problem which we call
{\em maxima with witnesses\/}:
given a point set $S$, report (the indices of) all maximal points in
left-to-right order, and for each nonmaximal point $p$ in $S$, 
report a maximal point (a {\em witness\/}) that dominates $p$.

\begin{theorem}\label{thm:maxima:alt}
$\OPT(S),\OPT\avg(S)\in\Omega(\FF(S)+n)$ for the 2-d ``maxima with witness'' problem in the
comparison model.
\end{theorem}
\begin{proof}
The proof is a counting argument, which we express in terms of encoding schemes (see \cite{DemaineLopezOrtizJALG03,GolynskiSODA09} for more sophisticated examples of counting arguments based on encoding/decoding).  We will describe a way to encode an arbitrary permutation $\sigma$ of $S$, so that the length of the encoding can be upper-bounded in terms of the running time of the given algorithm $A$ on input $\sigma$.  Since the worst-case encoding length must be at least $\log(n!)$, the running time must be large for some permutation $\sigma$.  (All logarithms are in base~2.)

To describe the encoding scheme, we imagine that the permutation $\sigma$ is initially unknown, and as we proceed, we record bits of information about $\sigma$ so that at the end, $\sigma$ can be uniquely determined from these bits.  In the description below, we distinguish between an {\em input point\/}, as represented its index in the input permutation $\sigma$ (its actual location is not necessarily known), and an {\em actual point\/} in $S$ (its coordinates are known but its index in $\sigma$ is not necessarily known).  At any moment, if we know which input point corresponds to an actual point $p$, we say (naturally) that $p$ is {\em known}.

We first simulate the algorithm on $\sigma$ and record the outcomes of the comparisons made; this requires $T_A(\sigma)$ bits (recall that $T_A(\sigma)$ denote the number of comparisons
made by $A$ on $\sigma$).
Let $M$ be the list of maximal input points returned.  For each input point $q_i$, let $W(q_i)$ be the list of all nonmaximal input points that have $q_i$ as witness.  For each maximal actual point, we record its position in $M$, using $h \up{\log h}$ bits in total.  Now all maximal points are known.

We process the nonmaximal actual points of $S$ from left to right, and make them known as follows.  To process an actual point $p$, let $q_i\TO q_j$ be all the maximal points that dominate $p$, which are all known.  Observe that $p$ must be in $W(q_i)\cup\cdots\cup W(q_j)$.  Let $L$ be all the points that are left of $p$, which are all known.  We record the position of $p$ in the list $(W(q_i)\cup\cdots\cup W(q_j)) - L$ of input points
(say, ordered by their indices).  This requires $\up{\log (|(W(q_i)\cup\cdots\cup W(q_j)) - L|}$ bits.  Observe that $W(q_i)\cup\cdots\cup W(q_j)$ is contained in $(-\infty,q_j.x)\times\R$.  So, $(W(q_i)\cup\cdots\cup W(q_j)) - L$ is contained in the subset $F(p)$ from our Definition~\ref{def:F}---a lucky coincidence.
   %
Thus, the number of bits required is $\up{\log|F(p)|}$.  Now $p$ is known and we can continue the process.

By our construction, any permutation $\sigma$ of $S$ can be uniquely decoded from its encoding, for any given set~$S$.
The encoding has total length at most
   $$ T_A(\sigma) + h\log h + \sum_p \log|F(p)| + O(n)
   \ =\ T_A(\sigma) + h \log h + n \log n - \FF(S) + O(n).$$ 
Taking the maximum over all permutations $\sigma$ of $S$,
we thus obtain $\log (n!)\ \le\ T_A(S) + h \log h + n \log n - \FF(S) + O(n)$, yielding $T_A(S) + n + h\log h \in \Omega(\FF(S))$.  Combined with the trivial lower bound $\Omega(n)$ and the naive information-theoretic lower bound $T_A(S) \in \Omega(h \log h)$ (as the problem definition requires the output to be in sorted order), this implies that $T_A(S) \in \Omega(\FF(S)+n)$.

The proof works in the random-order setting as well:
In any encoding scheme,
at most a fraction $2^{-cn}$ of the $n!$ permutations can have encoding length less than $\log (n!) - cn$ for any constant $c$.  Thus, for a random permutation $\sigma$, with high probability 
the encoding length is at least
$\log(n!)-O(n)$, implying $T_A(\sigma)\in\Omega(\FF(S)+n)$.
In particular, $T_A\avg(S)\in\Omega(\FF(S)+n)$.
\end{proof}

Combining the above theorem with the following observation yields a complete proof of the $\Omega(\FF(S)+n)$ lower bound:

\begin{observation}\label{obs:maxima:alt}
Any algorithm for the 2-d maxima problem in the comparison model can be made to solve the 2-d ``maxima with witnesses'' problem without needing to make any additional comparisons.
\end{observation}
\begin{proof}
Consider the partial order $\prec_x$ over $S$ formed by the outcomes of the $x$-comparisons made by the maxima algorithm~$A$.  Define the partial order $\prec_y$ similarly.
Fix a nonmaximal point $p$.  We argue that there must be a point $q\in S$ such that $p \prec_x q$ and $p \prec_y q$.  Suppose that
every $q\in S$ has $p\not\prec_x q$ or $p\not\prec_y q$.
Consider modifying the point set as follows:
First, increase the $x$-coordinates of $p$ and all points in 
$\{q\in S: p\prec_x q\}$ by a sufficiently large common value; this does not
affect the outcomes of the comparisons made, and ensures that
all points $q$ with $p\not\prec_x q$ now have $p.x > q.x$.
Similar, increase the $y$-coordinates of $p$ and all points in 
$\{q\in S: p\prec_y q\}$ by a sufficiently large common value; 
all points $q$ with $p\not\prec_y q$ now have $p.y > q.y$.
Then every $q\in S$ now has $p.x>q.x$ or $p.y>q.y$, i.e.,
$p$ is now maximal, and the algorithm would be incorrect on the modified point set: a contradiction.

For every nonmaximal point $p$, we can thus find a witness point $q$ that dominates~$p$, without making any additional comparisons.
One issue remains: the witness point may not be maximal.  If not, we can change $p$'s witness to the witness of the witness, and
so on, until $p$'s witness is maximal.

Finally, note that we do not require the given algorithm~$A$ to report the maximal points in left-to-right order.  We argue that
at the end we already know the $x$-order of the maximal points.
Suppose that $q\not\prec_x q'$ for two consecutive maximal points $q$ and $q'$.  Consider modifying the point set as follows:
Increase the $x$-coordinates of $q$ and
all points in $\{p\in S: q\prec_x p\}$ by a sufficiently large common value;
this does not affect the outcomes of the comparisons made, and
ensures that we now have $q.x > q'.x$ (while maintaining
$q.y > q'.y$).  Then $q'$ is now nonmaximal, and the algorithm would be incorrect on the modified point set: a contradiction.
\end{proof}

\begin{remark}
The proof can be modified
for weaker versions of the problem, for example,
reporting just the number of maximal points (or its parity).

The proof does not appear to work for problems other than maxima in 2-d.  One obvious issue is that Observation~\ref{obs:maxima:alt} only applies to comparison-based algorithms for orthogonal-type problems.  Even more critically, the proof of Theorem~\ref{thm:maxima:alt} relies on a coincidence that is special to 2-d maxima.

Curiously, this lower-bound proof holds even for nondeterministic algorithms, i.e., algorithms that can make guesses but must verify that the answer is correct; here we assume that each bit guessed costs unit time.
In the proof of Theorem~\ref{thm:maxima:alt}, we just record the guesses in the encoding.  The previous proofs of instance optimality by \citeN{Fagin} and \citeN{DemaineLopezOrtizMunro} all hold in nondeterministic settings.
Perhaps this strength of the proof prevents its applicability to other geometric problems, whereas our adversary-based proofs more powerfully exploits the deterministic nature of the algorithms.
\end{remark}


\bibliographystyle{ACM-Reference-Format-Journals}
\bibliography{abc_journal}

\end{document}